\documentclass{file}
\usepackage{pat}
\usepackage{paralist}
\usepackage{dsfont}  % for \mathds{A}
\usepackage[utf8x]{inputenc}
\usepackage{skull}
\usepackage{paralist}
\usepackage{centernot,enumitem}
\usepackage{mathtools}
\usepackage{stmaryrd}
\usepackage{algorithm,algpseudocode}
\usepackage{tcolorbox}
\usepackage[table]{xcolor}
\usepackage{graphicx}
\usepackage{bbm}  % needed for \mathbbm{1}
\usepackage{listings}
\usepackage{pifont}%
\usepackage{thm-restate}
\usepackage[normalem]{ulem}
\usepackage{cancel}

\usepackage{todonotes}
\usepackage{etoolbox}

\usepackage{float}
\usepackage{tabularx}
\usepackage{array}
\usepackage{subfig}
\usepackage{graphicx}
\usepackage{multirow}
\usepackage{makecell}

\newtheorem{observation}{Observation}
\newtheorem{definition}{Definition}
\newtheorem{theorem}{Theorem}
\newtheorem{lemma}{Lemma}

\newtheorem{notation}{Notation}
\newtheorem{example}{Example}
\newtheorem{claim}{Claim}
\newtheorem{corollary}{Corollary}

\newtheorem{remark}{Remark}
\definecolor{maroon}{rgb}{0.5, 0.0, 0.0}
\definecolor{darkblue}{rgb}{0.0, 0.0, 0.55}

\newcommand{\greductiontwo}{\mathcal{G}}
\newcommand{\vreductiontwo}{V(\greductiontwo)}
\newcommand{\ereductiontwo}{E(\greductiontwo)}
\newcommand{\np}{\mathbb{NP}}
\newcommand{\p}{\mathbb{P}}
\newcommand{\numitems}{|\mathbb X|}
\newcommand{\numsets}{|\mathbb F|}

\newcommand{\rr}{\epsilon_r}
\newcommand{\mm}{\mathsf{max}}

\newtheorem{iterationelltemp}{Iteration}

\newenvironment{clmproof}{\begin{proof}[Proof of Claim:]}{\end{proof}}
\usepackage[longnamesfirst,numbers,sort&compress]{natbib}

\usepackage[mathlines]{lineno}
\newcommand*\patchAmsMathEnvironmentForLineno[1]{%
 \expandafter\let\csname old#1\expandafter\endcsname\csname #1\endcsname
 \expandafter\let\csname oldend#1\expandafter\endcsname\csname end#1\endcsname
 \renewenvironment{#1}%
    {\linenomath\csname old#1\endcsname}%
    {\csname oldend#1\endcsname\endlinenomath}}%
\newcommand*\patchBothAmsMathEnvironmentsForLineno[1]{%
 \patchAmsMathEnvironmentForLineno{#1}%
 \patchAmsMathEnvironmentForLineno{#1*}}%
\AtBeginDocument{%
\patchBothAmsMathEnvironmentsForLineno{equation}%
\patchBothAmsMathEnvironmentsForLineno{align}%
\patchBothAmsMathEnvironmentsForLineno{flalign}%
\patchBothAmsMathEnvironmentsForLineno{alignat}%
\patchBothAmsMathEnvironmentsForLineno{gather}%
\patchBothAmsMathEnvironmentsForLineno{multline}%
}
\newcommand{\mi}{\mathsf{min}}

\newcommand{\alg}{\operatorname{ALG}}
\newcommand{\opt}{\mathsf{OPT}}
\def \G {\mathcal{G}_{tree}}

\usepackage{accents}
\newlength{\dhatheight}
\newcommand{\doublehat}[1]{%
    \settoheight{\dhatheight}{\ensuremath{\hat{#1}}}%
    \addtolength{\dhatheight}{-0.05ex}%
    \hat{\vphantom{\rule{2pt}{\dhatheight}}%
    \smash{\hat{#1}}}}

\definecolor{brightmaroon}{rgb}{0.76, 0.13, 0.28}
\definecolor{linkblue}{rgb}{0, 0.337, 0.227}
\makeatletter
\newcommand{\xMapsto}[2][]{\ext@arrow 0599{\Mapstofill@}{#1}{#2}}
\def\Mapstofill@{\arrowfill@{\Mapstochar\Relbar}\Relbar\Rightarrow}
\makeatother

\title{{Triangle-Covered Graphs: Algorithms, Complexity, and Structure}}%\thanks{This research was partly funded by NSERC.}}
\author{
Amirali Madani\thanks{School of Computer Science, Carleton University, Ottawa, Ontario, Canada. Supported by the Natural Sciences and Engineering Research Council of Canada (NSERC). Emails: {\tt amiralimadani@cmail.carleton.ca} and {\tt anil@scs.carleton.ca}}
\;\;\; Anil Maheshwari\footnotemark[2] \;\;\; Bobby Miraftab\thanks{School of Computer Science, Carleton University, Ottawa, Ontario, Canada. Email: {\tt bobby.miraftab@gmail.com}}
\;\;\; Paweł Żyliński \thanks{ Institute of Informatics, University of Gdańsk, Gdańsk, Poland Email: {\tt pawel.zylinski@ug.edu.pl}}
}

\date{}

\begin{document}

\maketitle
\begin{abstract}
The widely studied edge modification problems ask how to minimally alter a graph to satisfy certain structural properties. 
In this paper, we introduce and study a new edge modification problem centered around transforming a given graph into a \defin{triangle-covered} graph (one in which every vertex belongs to at least one triangle). 
We first present tight lower bounds on the number of edges in any connected triangle-covered graph of order $n$ and then we characterize all connected graphs that attain this minimum edge count.
For a graph $G$, we define the notion of a \defin{$\Delta$-completion set} as a set of non-edges of $G$ whose addition to $G$ results in a triangle-covered graph. 
We prove that the decision problem of finding a $\Delta$-completion set of size at most $t\geq0$ is $\mathbb{NP}$-complete and does not admit a constant-factor approximation algorithm under standard complexity assumptions. Moreover, we show that this problem remains $\mathbb{NP}$-complete even when the input is restricted to connected bipartite graphs. We then study the problem from an algorithmic perspective, providing tight bounds on the minimum $\Delta$-completion set size for several graph classes, including trees, chordal graphs, and cactus graphs. Furthermore, we show that the triangle-covered problem admits an $(\ln n +1)$-approximation algorithm for general graphs. For trees and chordal graphs, we design algorithms that compute minimum $\Delta$-completion sets.
Finally, we show that the threshold for a random graph $\mathbb{G}(n, p)$ to be triangle-covered occurs at $n^{-2/3}$.  
\end{abstract}
\section{Introduction}
\label{sec:intro}

The field of \emph{edge modification problems} has seen considerable interest in recent years due to its wide applicability across various fields in graph theory. Roughly speaking, these problems involve transforming a given graph $G$ into another graph $G'$ 
that satisfies a particular property, typically by making the minimal number of edge insertions. Although the exact nature of these problems varies, a common formulation seeks to determine how few edges need to be added (or sometimes removed) to enforce a target structural feature, see the survey~\cite{edgemodificationsurvey} for a comprehensive review of this area.

The motivation for this work arises from a particular blend of edge modification and community search, the latter being a core problem in the analysis of large networks and social graphs. In this context, a key goal is to alter a graph in a minimal way so that it reveals underlying community structures more clearly. A well-studied example involves modifying graphs into cluster graphs, where each connected component is a clique~\cite{cluster4, cluster3, cluster2, cluster1}. Such transformations are particularly useful for identifying tightly-knit groups within a network, whether in social science, biology, or recommendation systems~\cite{clusterapplication}.
In the study of community search, cohesiveness plays a central role. Various definitions of cohesiveness exist. One common approach considers a subgraph cohesive if it maintains a certain minimum degree~\cite{corekeywordcommunitysearch}, while another requires that every edge participates in at least one triangle~\cite{trusskeywordcommunitysearch}. These criteria serve as the foundation for many algorithms designed to detect communities or dense substructures. Recent work by Fomin, Sagunov, and Simonov~\cite{fomin2023building}, along with Chitnis and Talmon~\cite{chitnis}, examines how to strategically add edges to an input graph $G$ so that it becomes cohesive under such definitions.

Graphs with local covering conditions, which require every edge or vertex to be contained within specific substructures, have gained considerable attention in extremal graph theory. 
For instance, Burkhardt, Faber, and Harris~\cite{BurkhardtFH20} derived sharp asymptotic lower bounds on the edge count in connected graphs where each edge is contained in at least  $\ell$  distinct triangles, see references \cite{chakraborti2024sparse,DBLP:conf/caldam/MadaniMMR25,DBLP:journals/corr/abs-2502-02572}.
In this paper, we are interested in the following local condition in which every vertex lies in a triangle.
\begin{defn}
A graph $G$ is called \defin{triangle-covered} if every vertex of $G$ belongs to at least one triangle. 
\end{defn}
Chakraborti and Loh \cite{chakraborti2020extremal} determined the minimum edge count and characterized the corresponding extremal graphs for triangle-covered graphs, where disconnected graphs are allowed. 
Here, an \emph{extremal} graph means one whose number of edges attains the graph-theoretic lower bound.
We note that the condition of ``connectedness'' completely changes the problem.  
For example, when $ n = 7 $, there is only one extremal triangle-covered graph, but it is not an extremal connected triangle-covered graph, see \cref{fig:dis_7}.
\begin{figure}[H]
    \centering
    \includegraphics[scale=0.8]{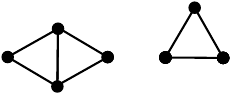}
    \caption{An example of an extremal triangle-covered graph, but it is not  extremal connected triangle-covered.}
    \label{fig:dis_7}
\end{figure}
However, one can join these two subgraphs by an edge to obtain an extremal connected triangle-covered graph, and this can be done in two different ways.
But these are not the only extremal examples on 7 vertices in the connected case, see \cref{fig:connected_7}.
\begin{figure}[H]
    \centering
    \includegraphics[scale=0.7]{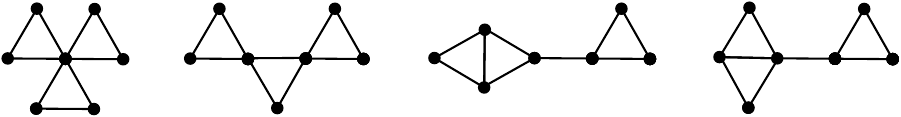}
    \caption{Extremal connected
triangle-covered graphs on $7$ vertices.}
    \label{fig:connected_7}
\end{figure}
Therefore, not every extremal connected triangle-covered graph can be obtained by connecting components of an extremal (possibly disconnected) triangle-covered graph.
In this paper, we determine the minimum number of edges in a connected triangle-covered graph of order $n$.

\begin{restatable}{theorem}{mainone}
\label{main_1}
Let $G$ be a connected triangle-covered graph of order $n$, where $n=3q+r>3$ and $0\leq r<3$.
Then $|E(G)|\geq 4q-1+\rr$, where $\rr\in\{0,2,3\}$.
\end{restatable}

In the following example, we construct a connected triangle-covered graph of order $n$ with the minimum number of edges (for every $ n\geq 3 $), and we can conclude that the lower bound stated in \cref{main_1} is tight.

\begin{example}
Let $ n $ be an integer such that $ n = 3q + r $ and $ 0 \leq r < 3 $. Then, there exists a connected triangle-covered graph with $ n $ vertices and $ 4q - 1 + \epsilon_r $ edges, where $\epsilon_r\in \{0,2,3\}$.
To construct such a graph, start with $ q $ triangles connected by $ q - 1 $ edges, as illustrated in \cref{fig:tri_0} (a).
If $ n > 3q $, proceed as follows: 
\begin{itemize}
    \item If $ n = 3q + 1 $, then add the last vertex by connecting it with the endpoints of an edge, as shown in \cref{fig:tri_0} (b). 
    \item If $ n = 3q + 2 $, then first join the two remaining vertices with an edge. Then, connect each of these vertices to one of the existing vertices, as shown in \cref{fig:tri_0}(c).
\end{itemize}
\begin{figure}[H]
    \centering
    \subfloat[\centering ]{{\includegraphics[scale=0.65]{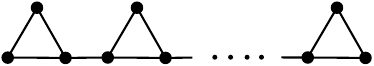}}}
\qquad
\subfloat[\centering ]{{\includegraphics[scale=0.65]{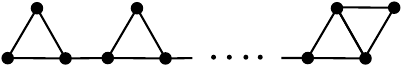} }}%
\qquad
\subfloat[\centering ]{{\includegraphics[scale=0.65]{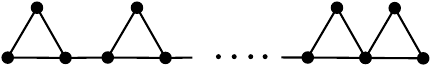} }}%
    \caption{Examples of a triangle-covered  graphs}
    \label{fig:tri_0}
\end{figure}
\end{example}

Furthermore, we characterize all connected graphs that attain this minimum edge count.
More precisely, we prove that every such graph admits a  decomposition in a tree-like way where each bag has a specific structure; see \Cref{thm:main2} for details.
In this paper, we also explore the algorithmic aspects of connected triangle-covered graphs. 

\begin{definition}
Let $G$ be a connected graph. A subset $ F $ of non-edges of $G$ is called a \defin{$\Delta$-completion set} of $G$
if the graph obtained by adding $ F $ to $ G $ becomes triangle-covered.  
\end{definition}

The first algorithmic question regarding triangle-covered graphs is as follows.

\begin{description}
\item[The triangle-covered problem.] Given a graph $ G $ and an integer $ t \geq 0 $, does $ G $ have a $\Delta$-completion set of size at most $ t $? \end{description}

We prove that this problem is $\np$-complete for general graphs by providing a polynomial-time reduction from the set cover problem. From our hardness reduction, we can also conclude that this problem does not admit a constant-factor approximation algorithm running in polynomial time, unless $ \p = \np $.

\begin{restatable}{theorem}{mainthree}
\label{main_3}
The triangle-covered problem is $\mathbb{NP}$-complete. Moreover, there exists no polynomial-time $c$-approximation algorithm for the triangle-covered problem for any constant $c>1$, unless $\p=\np$. 
\end{restatable}

However, we show that this problem admits an $(\ln n+1)$-approximation algorithm, where $n$ is the order of the input graph. The next algorithmic question regarding triangle-covered graphs is as follows. 

\begin{description}
\item[Minimum triangle-covered problem.]
Let $ \mathcal{C} $ be a class of connected graphs.  
What is the minimum number of non-edges that must be added to each graph in $ \mathcal{C} $ to ensure that the resulting graph is \textit{triangle-covered}? Additionally, can we design an optimal algorithm to achieve this minimum number of edge additions?
\end{description}

\noindent 
We denote the size of a minimum $\Delta$-completion set of $G$ by $\Delta_G$.
We present lower and upper bounds for $ \Delta_G $ in certain classes of graphs.  
For trees and chordal graphs, we provide an efficient algorithm to compute a minimum $ \Delta $-completion.
Below is a summary of our results.
\begin{center}
        \rowcolors{2}{gray!30}{white}
\begin{table}[H]
  \centering
    \caption{Summary of our results.}
  \resizebox{\textwidth}{!}{
%   \begin{tabular}{l c c c} \hline
    \begin{tabular}{p{6cm} c c c} \hline
    {\bf Graph Class} & Lower Bound & Upper Bound & Reference \\ \hline
    Paths & $\lceil\frac{n}{3}\rceil$ & $\lceil\frac{n}{3}\rceil$ & - \\ %\hline
    Stars & $\lfloor\frac{n}{2}\rfloor$ & $\lfloor\frac{n}{2}\rfloor$ & - \\ %\hline
    Double stars & $\lceil\frac{n}{2}\rceil$-1 & $\lceil\frac{n}{2}\rceil$ & \Cref{prop:double_star_exa} \\ %\hline
    Trees & $\lceil\frac{n}{3}\rceil$ & $\lceil\frac{n}{2}\rceil$ & \Cref{thm:tree} \\ %\hline 
    2-edge-connected Cactus & 0 & $\lceil\frac{2n}{5}\rceil$ & \Cref{thm:2_conn_cactus} \\ \hline 
    \end{tabular}
  }
  \label{summary-table}
\end{table}
\end{center}

Lastly, we also discuss the threshold at which random graphs become triangle-covered.

\begin{restatable}{theorem}{mainfour}
\label{main_4}
Let $ G = \mathbb G(n, p) $ be a random graph.
Then  $$ p \gg n^{-2/3} \textit{ if and only if }\lim_{n\to \infty}\mathbb P(G \textit{ is a triangle-covered graph})=1$$
\end{restatable}

\subsection{Overview and Organization of the Paper}
The remainder of this paper is organized as follows. 
We first establish the tight edge lower bound for connected triangle-covered graphs and give matching constructions (\Cref{main_1}). 
We then characterize all extremal connected triangle-covered graphs via a tree-like bag decomposition (\Cref{thm:main2}). 
Next, we prove a local structural lemma showing that a minimum $\Delta$-completion can be chosen using only distance-two edges and derive separation consequences (\Cref{lem:lower_bound,cor:lower_bound_for_subset,dist=3}). 
Our complexity results show $\np$-completeness and inapproximability (\Cref{main_3}) and that hardness persists on connected bipartite inputs; a reduction to \textsc{Set-Cover} yields a greedy $(\ln n+1)$-approximation. 
We then present tight bounds and optimal algorithms: for trees (\Cref{thm:tree,thm:alg-trees}), for chordal graphs  (\Cref{thm:chordal}), and for $2$-edge-connected cactus graphs with the tight bound $\lceil 2n/5\rceil$ (\Cref{thm:2_conn_cactus}). Finally, we locate the threshold for $\mathbb G(n,p)$ to be triangle-covered at $p=n^{-2/3}$ (\Cref{main_4}) and conclude with open questions.

%%%%%%%%%%%%%%%%%%%%%%%%%%%%%%%%%%%%%%%%%%%%%%%%%%%%%%%%%%%%%%%%%%%%%%%%%%%%%%%%%%%%%%%%%%%%%%%%%%%%%%%%%%%%%%%%%%%%%%%%%%%%%%%%%%%%%%%%%%%%%%%%%%%%%%%%%%%%%%%%%%%%%%%%%%%%%%%%%%%%%%%%
\section{Tight Lower Bounds on the Number of Edges}
In this section, we determine the minimum number of edges in any connected triangle-covered graph with $n$  vertices. 
We start with the following useful notation.
\begin{notation}\label{notause}
Let $r$ be an integer such that $0\leq r<3$.
Then, $\rr = \frac{r(5 - r)}{2}$.
\end{notation}
\begin{remark}\label{r}
We note that, with \cref{notause}, $r\neq 0$ if and only if $\rr - r=1$.
\end{remark}
The problem of covering edges with triangles has already been studied, and we aim to leverage those results for our problem.
A graph $G$ has a \defin{$(3,1)$-edge-cover} if every edge of $G$ lies in a triangle. 
%\todo{anil: My suggestion will be to call it $(3,1)$-edge cover} 
Chakraborti et al.~\cite{chakraborti2024sparse} proved the following.

\begin{lemma}{\em\cite{chakraborti2024sparse}}
Let $G$ be a connected $n$-vertex graph with a \mbox{$(3,1)$-edge-cover}. 
Then, the number of edges in $G$ is at least $\frac{3(n-1)}{2}$. 
\end{lemma}

\begin{remark}\label{rmk:previous_discussion}
It is worth mentioning that if $ n = 3q + r \geq 8 $, where $ 0 \leq r < 3 $, then  
$\frac{3(n-1)}{2} > 4q - 1 + \epsilon_r$.
Because
if every edge were in a triangle, then $|E(G)|\ge \tfrac{3(n-1)}{2} > 4q-1+\epsilon_r$ for$ n\ge 8$, contradicting $|E(G)|<4q-1+\epsilon_r$.
This implies that an extremal connected triangle-covered graph cannot have a $(3,1)$-edge-cover and must contain at least one edge that does not participate in any triangle.
\end{remark}

\mainone*

\begin{proof}
It follows from \cref{fig:blobs2}  that we can assume that $n\geq 8$.
Assume, for the sake of contradiction, that $ G $ is the smallest counterexample such that $ |E(G)| < 4q - 1 + \epsilon_r $ and $ |V(G)| = n $.
Moreover, we can assume that the number of its edges is the minimum possible.
It follows from \Cref{rmk:previous_discussion} that there exists an edge $ e $ of $ G $ that does not belong to any triangle.  
This edge must be a bridge; otherwise, the graph obtained by deleting $ e $ would still be a triangle-covered graph, which contradicts the assumption that $ G $ has the minimum number of edges possible.
Let $ G_1 $ and $ G_2 $ be the two components of $ G \sm e $ with $ n_1 $ and $ n_2 $ vertices, respectively.
We note that each $G_i$ is a triangle-covered graph.
Let $n_1=3q_1+r_1$, we have the following two cases.
\begin{enumerate}
\item $n-n_1=3(q-q_1) +r-r_1$ if $r\geq r_1$.
\item $n-n_1=3(q-q_1-1) +r-r_1+3$ if $r<r_1$.
\end{enumerate}
Observe that $|E(G_1)| + |E(G_2)| + 1 < 4q - 1 + \rr$.  
We have two cases based on whether $r - r_1$ is positive or not. \\[2mm] 
\noindent {\bf Case I:} Assume that $r \geq r_1$. Then, one can see that  
\[
\underbrace{4q_1 - 1+\epsilon_{r_1}}_{|E(G_1)|}   + \underbrace{4(q - q_1) - 1 + \epsilon_{r - r_1}}_{|E(G_2)|} + 1 < \underbrace{4q - 1 + \rr}_{|E(G)|},
\]
which implies that $\epsilon_{r_1} + \epsilon_{r - r_1} < \rr$.  
If $r = 0$ or $r = r_1$, we get a contradiction.  
So, assume that $r \neq r_1$ and also $r \neq 0$.  
By \cref{r}, we know that $\epsilon_{r_1} + r - r_1 + 1 < 1 + r$.  
In other words, we have $\epsilon_{r_1} - r_1 < 0$, which yields a contradiction according to \cref{r}.\\[2mm] 
\noindent {\bf Case II:} Assume that $r < r_1$. Then, one can see that  
\[
\underbrace{4q_1 - 1 + \epsilon_{r_1}}_{|E(G_1)|} + \underbrace{4(q - q_1 - 1) - 1 + \epsilon_{r - r_1 + 3}}_{|E(G_2)|} + 1 < \underbrace{4q - 1 + \rr}_{|E(G)|},
\]  
which implies that $\epsilon_{r_1} + \epsilon_{r - r_1 + 3} - 4 < \rr$.  
If $r_1 = 1$, we get $\epsilon_{1} + \epsilon_2 - 4 < \epsilon_0$, which yields a contradiction.  
Otherwise, $r_1 = 2$, which implies that $\epsilon_{2} + \epsilon_{r + 1} - 4 < \epsilon_{r}$.  
If $r = 0$, then $\epsilon_{2} + \epsilon_{1} - 4 < \epsilon_{0}$, which yields a contradiction.  
Otherwise, $r = 1$, and so $\epsilon_{2} + \epsilon_{2} - 4 < \epsilon_{1}$, which also results in a contradiction.  
\end{proof}
\begin{figure}[H]
    \centering
\subfloat[\centering $H_1$ ]{{\includegraphics[scale=0.75]{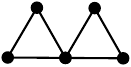} }}%
\qquad
\subfloat[\centering $H_2$]{{\includegraphics[scale=0.75]{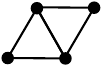} }}%
\qquad
\subfloat[\centering $H_3$ ]{{\includegraphics[scale=0.75]{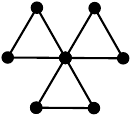} }}%
\qquad
\subfloat[\centering $H_4$ ]{{\includegraphics[scale=0.75]{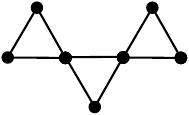} }}%
\qquad
\subfloat[\centering $H_5$ ]{{\includegraphics[scale=0.75]{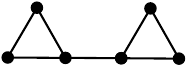} }}%
\qquad
\subfloat[\centering $H_6$ ]{{\includegraphics[scale=0.75]{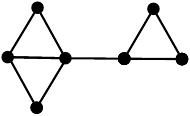} }}%
\qquad
\subfloat[\centering $H_7$ ]{{\includegraphics[scale=0.75]{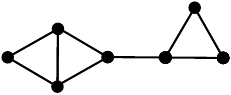} }}%
    \caption{All connected triangle-covered graphs on $4$, $5$, $6$, and $7$ vertices with the minimum possible number of edges are shown. 
Graph $H_1$ has order $5$, $H_2$ has order $4$, $H_5$ has order $6$, and $H_3$, $H_4$, $H_6$, and $H_7$ each have order $7$. 
These graphs were found using SageMath~\cite{sage}.  }
    \label{fig:blobs2}
\end{figure}
%%%%%%%%%%%%%%%%%%%%%%%%%%%%%%%%%%%%%%%%%%%%%%%%%%%%%%%%%%%%%%%%%%%%%%%%%%%%%%%%%%%%%%%%%%%%%%%%%%%%%%%%%%%%%%%%%%%%%%%%%%%%%%%%%%%%%%%%%%%%%%%%%%%%%%%%%%%%%%%%%%%%%%%%%%%%%%%%%%%%%%%%
\section{Characterization of Extremal Graphs}

\begin{defn}
Let $ G $ be a connected triangle-covered graph on $ n = 3q + r $ vertices, where $ 0 \leq r < 3 $.  
We say that $ G $ is an \defin{extremal triangle-covered} graph if the number of edges is $ 4q - 1 + \rr $.  

\end{defn}
Next, we show that every extremal triangle-covered graph of order at least $9$ can be decomposed into blobs (bags) in a tree-like manner, where each bag is a triangle except for at most two bags. 
%\todo{Anil: We can say at most two}
From now on, for simplicity of presentation, for a natural number~$m$, we shall write $[m]$ to denote the set $\{1,\dots,m\}$.
Consider the family $\mathcal{H}=\{H_1, \ldots, H_7\}$ of seven graphs depicted in \cref{fig:blobs2}.

\begin{defn}\label{bags}
Define $\G(n)$ to be the family of graphs $G$ on $n$ vertices such that there is a tree $T$ with the vertex set $[m]$ and a family of \emph{disjoint} $(V_i)_{i\in [m]}$ of vertex sets~$V_i\subseteq V(G)$, one for every vertex of~$T$ satisfying the following. 
\begin{enumerate}
		\item $V(G) = \bigcup_{i\in [m]} V_i$, 
		\item Exactly one of the following occurs:
        \begin{enumerate}[label=\Alph* :]
         \item For every $ i \in [m] $, the graph $ G[V_i] $ induced by $ V_i $ is isomorphic to either $ H_1 $ or $ K_3 $ and the number of bags isomorphic to $H_1$ is at most $2$.
            \item For every $ i \in [m] $, except one $ j $, the graph $ G[V_i] $ induced by $ V_i $ is isomorphic to $ K_3 $, and $ G[V_j] $ is isomorphic to one of $ H_{\ell} $ for $ \ell = 2, 3, 4 $.
\end{enumerate}
\item For every edge $uv\in E(G)$, exactly one of the following occurs:
        \begin{itemize}
            \item there exists $i\in [m]$ such that $u,v\in V_i$,
            \item there are unique $i\neq j\in [m]$ such that $u\in V_i$ and $v\in V_j$ and $ij\in E(T)$.
        \end{itemize}
	\end{enumerate}
Every $V_i$ is called a \defin{bag}.
\end{defn}

\begin{exa}\label{<7}
With the help of SageMath {\rm\cite{sage}}, we found all extremal connected triangle-covered graphs with 8 vertices. The dashed lines indicate the corresponding bags. Here $\G(8)$ has only two members and each one has only two bags.
\begin{figure}[H]
\centering
\subfloat[]{{\includegraphics[scale=0.75]{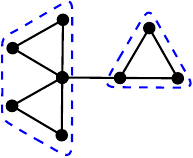} }}%
\qquad
\subfloat[ ]{{\includegraphics[scale=0.75]{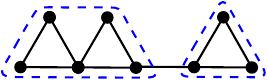} }}%
    \caption{All extremal triangle-covered graphs with $8$ vertices.}
    \label{fig:blobs}
\end{figure}
\end{exa}
Next, we characterize all extremal triangle-covered graphs with at least 8 vertices. 

\begin{theorem}\label{thm:main2}
Let $G$ be a connected graph on $n\geq 8$ vertices.
Then $G$ is an extremal triangle-covered graph if and only if $G\in \G(n)$.
\end{theorem}

\begin{proof}
We first assume that $ G \in \mathcal{G}(n) $, where $ n \geq 8 $. We aim to show that $ |E(G)| = 4q - 1 + \epsilon_r $, where $ n = 3q + r $. Let the number of bags in $ G \in \mathcal{G}(n) $ be $ m $.
\begin{itemize}
    \item If one of the following conditions holds, then $ n = 3m + 4 $, $ |E(G)| = 4m + 5 $, and thus $ q = m + 1 $ and $ r = 1 $:
    \begin{itemize}
        \item There are two bags isomorphic to $ H_1 $.
        \item There is exactly one bag isomorphic to $ H_3 $ or $ H_4 $.
    \end{itemize}
    \item If there is exactly one bag isomorphic to either $ H_1 $ or $ H_2 $, then $ n = 3m + 2 $ and $ |E(G)| = 4m + 2 $. Thus, $ q = m $ and $ r = 2 $.
    \item If all bags are triangles, then $ n = 3m $, and the result holds. 
\end{itemize}

We now assume that $ G $ is an extremal triangle-covered graph.     
We verified all cases $n\le 8$ by exhaustive check (\Cref{<7} and \Cref{fig:blobs2}); hence assume $n\ge 9$.
It follows from \Cref{rmk:previous_discussion} that there exists an edge $ e $ of $ G $ that does not belong to any triangle.  
This edge must be a bridge.
Let $ G_1 $ and $ G_2 $ be the two components of $ G \sm e $ with $ n_1 $ and $ n_2 $ vertices, respectively.
\begin{clm}
Each $G_i$ is an extremal triangle-covered graph.
\end{clm}
\begin{proof}
Assume to the contrary that $G_1$ is not an extremal triangle-covered graph.
Let $n_1=3q_1+r_1$ and so we have the following two cases:
\begin{enumerate}
\item $n-n_1=3(q-q_1) +r-r_1$ if $r\geq r_1$.
\item $n-n_1=3(q-q_1-1) +r-r_1+3$ if $r<r_1$.
\end{enumerate}
We note that $|E(G_1)| + |E(G_2)| + 1 = 4q - 1 + \rr$.  
Since $G_1$ is not an extremal triangle-covered graph, we have two cases based on whether $r - r_1$ is positive or not. \\[2mm] 
\noindent {\bf Case I:} Assume that $r \geq r_1$. Then, one can see that  
\[
4q_1 - 1 + \epsilon_{r_1} + 4(q - q_1) - 1 + \epsilon_{r - r_1} + 1 < 4q - 1 + \rr,
\]  
which implies that $\epsilon_{r_1} + \epsilon_{r - r_1} < \rr$.  
If $r = 0$ or $r = r_1$, we get a contradiction.  
So, assume that $r \neq r_1$ and also $r \neq 0$.  
By \cref{r}, we know that $\epsilon_{r_1} + r - r_1 + 1 < 1 + r$.  
In other words, we have $\epsilon_{r_1} - r_1 < 0$, which yields a contradiction according to \cref{r}.\\[2mm] 
\noindent {\bf Case II:} Assume that $r < r_1$. Then, one can see that  
\[
4q_1 - 1 + \epsilon_{r_1} + 4(q - q_1 - 1) - 1 + \epsilon_{r - r_1 + 3} + 1 < 4q - 1 + \rr,
\]  
which implies that $\epsilon_{r_1} + \epsilon_{r - r_1 + 3} - 4 < \rr$.  
If $r_1 = 1$, we get $\epsilon_{1} + \epsilon_2 - 4 < \epsilon_0$, which yields a contradiction.  
Otherwise, $r_1 = 2$, which implies that $\epsilon_{2} + \epsilon_{r + 1} - 4 < \epsilon_{r}$.  
If $r = 0$, then $\epsilon_{2} + \epsilon_{1} - 4 < \epsilon_{0}$, which yields a contradiction.  
Otherwise, $r = 1$, and so $\epsilon_{2} + \epsilon_{2} - 4 < \epsilon_{1}$, which yields a contradiction.  
\end{proof}
If a component $C$ has more than 8 vertices, with a similar argument, $C$ should contain a bridge.
So we continue this process iteratively until every component contains at most 7 vertices.
Next, we demonstrate that the conditions of \cref{bags} are satisfied for each component and by our construction, $G\in \G(n)$.
Assume, for the sake of contradiction, that $ G $ has two bags $ B_1 $ and $ B_2 $ which are not triangles.  
We first introduce a notation. If $ G $ is a graph with a subgraph $ X $, then $ \hat{G}(X) $ is a graph obtained by contracting  $X$ to a vertex $x$(replacing all vertices and edges of 
$X$ with a single vertex $x$), and then contracting an edge incident to $x$ means merging $x$ with an adjacent vertex 
$y$ into a new vertex, see \cref{fig:con}.
We note that $ \hat{G}(X) $ is not unique, as it depends on the choice of an incident edge of $X$ to contract. 
\begin{figure}[H]
    \centering
    \includegraphics[scale=0.7]{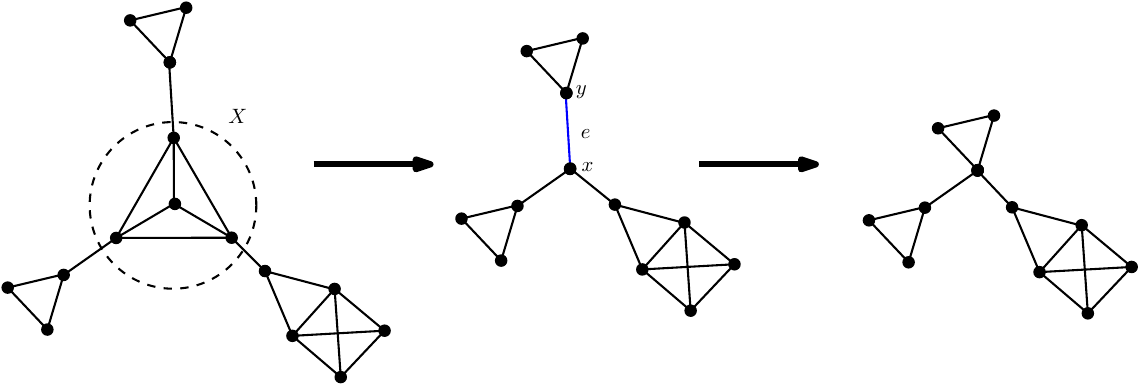}
    \caption{The graph on the left represents $ G $, with $ H $ denoting the subgraph depicted by the dashed lines and $ e $ being the contracted edge. The graph on the right illustrates $ \hat{G}(H) $.} 
    \label{fig:con}
\end{figure}
Similarly, one can define $ \doublehat{G}(X)(Y) $ for more subgraphs, which means we first contract $X$ with an incident edge and then contract $Y$ with an incident edge.
Furthermore, we show that our graph is large enough so that we can select two distinct incident (connecting) edges for contracting.
We consider the following cases based on the type of these bags:

\begin{enumerate}

\item Let $ B_1, B_2 \in \{H_3, H_4\} $. 
Assume first that $G$ is the union of $B_1$ and $B_2$ connected with an edge. 
So $|V(G)|=14$ and $|E(G)|=19$.
However, it follows from \cref{main_1} that any extremal example should have $18$ edges. 
So we conclude that $|V(G)|\geq 15$.
   We set $H\coloneqq \doublehat{G}(B_1)(B_2)$(or $\doublehat{G}(B_2)(B_1)$).
   We note that since the order of $G$ is at least 15, we are able to contract two distinct edges in $H$.
   Then the resulting graph $ H $ is still a triangle-covered graph, and we have $ |E(H)| \geq 4q' - 1 + \epsilon_{r'} $, where
    \[
    q' =
    \begin{cases}
    q - 4, & \text{if } r = 2, \\
    q - 5, & \text{if } r \neq 2,
    \end{cases}
    \quad
    r' =
    \begin{cases}
    0, & \text{if } r = 2, \\
    r + 1, & \text{if } r \neq 2.
    \end{cases}
    \]
    We note that when contracting a bag $B_1$ with $|V(B_1)|$ vertices and $|E(B_1)|$ edges,  the number of vertices reduced is $|V(B_1)| - 1$, and the number of edges lost is  
$|E(B_1)| + 1$ (including the connecting edge).  
For two bags, the total number of edges lost is $|E(B_1)| + |E(B_2)| + 2$.
Thus we have the following calculation:
    $
    |E(H)| + 2 + |E(B_1)| + |E(B_2)| = |E(G)| = 4q - 1 + \epsilon_r,
    $
    which implies that either  
    \[
    4q - 1 + \epsilon_2 \geq 20 + 4(q - 4) - 1 + \epsilon_0,
 \textit{ or } 
    4q - 1 + \epsilon_r \geq 20 + 4(q - 5) - 1 + \epsilon_{r+1}.
    \]
    Both cases lead to a contradiction.  
    \item Let $ B_1=B_2=H_2$. 
    If $ G $ is exactly two copies of $ H_2 $ connected by a single edge, then $ |E(G)| = 11 $ and $ |V(G)| = 8 $.  
However, by \Cref{main_1}, we know that any extremal example must have exactly $ 10 $ edges.  
Thus, $ |V(G)| \geq 9 $, and we can infer that $ H \coloneqq  \doublehat{G}(B_1)(B_2) $(or $\doublehat{G}( B_2)(B_1)$) is still a triangle-covered graph. Moreover, we have  
\[
|E(H)| \geq 4q' - 1 + \epsilon_{r'}.
\] where
    \[
    q' =
    \begin{cases}
    q - 2, & \text{if } r = 2, \\
    q - 3, & \text{if } r \neq 2,
    \end{cases}
    \quad
    r' =
    \begin{cases}
    0, & \text{if } r = 2, \\
    r + 1, & \text{if } r \neq 2.
    \end{cases}
    \]
    Next, note that 
    $
    |E(H)| + 2 + |E(H_1)| + |E(H_2)| = |E(G)| = 4q - 1 + \epsilon_r,
    $
    which implies that either  
    \[
    4q - 1 + \epsilon_2 \geq 12 + 4(q - 2) - 1 + \epsilon_0,
 \textit{ or } 
    4q - 1 + \epsilon_r \geq 12 + 4(q - 3) - 1 + \epsilon_{r+1}.
    \]
    Both cases lead to a contradiction.
\item Let $ B_1=H_2$ and $B_2\in \{H_3,H_4\}$. 
If $G$ is the union of $B_1$ and $B_2$ connected by an edge, then $|E(G)|=16$.
However, by \Cref{main_1} we know that any extremal example of order $11$ should have $15$ edges.
So we can assume that $|V(G)|\geq 12$ and we can verify that $ H\coloneqq\doublehat{G}( B_1)(B_2)$(or $ \doublehat{G}( B_2)(B_1)$) is still a triangle-covered graph, and we have $ |E(H)| \geq 4q' - 1 + \epsilon_{r'} $, where
    \[
    q' =
    \begin{cases}
    q - 4, & \text{if } r = 0, \\
    q - 3, & \text{if } r \neq 2,
    \end{cases}
    \quad
    r' =
    \begin{cases}
    2, & \text{if } r = 0, \\
    r - 1, & \text{if } r \neq 0.
    \end{cases}
    \]
    Next, note that 
    $
    |E(H)| + 2 + |E(H_1)| + |E(H_2)| = |E(G)| = 4q - 1 + \epsilon_r,
    $
    which implies that either  
    \[
    4q - 1 + \epsilon_2 \geq 16 + 4(q - 3) - 1 + \epsilon_0,
 \textit{ or } 
    4q - 1 + \epsilon_r \geq 16 + 4(q - 3) - 1 + \epsilon_{r-1}.
    \]
    %\noindent\textbf{\color{red}PZ:} {\color{blue} $\epsilon_{r+1}$ or $\epsilon_{r-1}$?}
    
    Both cases  lead  to a contradiction. 
    \item Let $B_1=H_1,B_2\in \{H_3,H_4\}$. With an analogous method of the previous cases, we can show that 
    $4q - 1 + \epsilon_r \geq 17 + 4(q - 4) - 1 + \epsilon_r$, which yields a contradiction. 
    \item Let $ B_1 = B_2 = B_3 = H_1 $.  
Let $ G $ be the union of $ B_1, B_2, $ and $ B_3 $, connected with only two edges.  
Then $ |V(G)| = 15 $ and $ |E(G)| = 18 $.  

By \Cref{main_1}, we know that any extremal example should have 19 edges.  
Thus, we can assume that $ |V(G)| \geq 16 $ and define  
$H \coloneqq\hat{G}^{(3)}(B_i, B_j, B_k)
$, 
%\todo{Anil: Should we explain what this means?}  
where $ i, j, k $ are distinct elements of $ \{1,2,3\} $.
In other words we first have $H_1=\hat{G}(B_i)$ and then $H_2=\hat H_1(B_j)$ which implies that $H=\hat{H_2}(B_k)$.
    One can see that $ H$ is still a triangle-covered graph, and we have $ |E(H)| \geq 4q' - 1 + \epsilon_{r'} $, where $q'=q-5$ and $r'=r$.
    Then, in this case, we have 
    $4q - 1 + \epsilon_r \geq 21 + 4(q - 5) - 1 + \epsilon_r$, which yields a contradiction.%\qedhere
\end{enumerate}
\end{proof}
%%%%%%%%%%%%%%%%%%%%%%%%%%%%%%%%%%%%%%%%%%%%%%%%%%%%%%%%%%%%%%%%%%%%%%%%%%%%%%%%%%%%%%%%%%%%%%%%%%%%%%%%%%%%%%%%%%%%%%%%%%%%%%%%%%%%%%%%%%%%%%%%%%%%%%%%%%%%%%%%%%%%%%%%%%%%%%%%%%%%%

\section{Structural Properties}

In this section, we prove a useful property of minimum $\Delta$-completion sets that will be used throughout the paper.
We first present some notation.
In this section, we drop the assumption of connectedness.
Let $G=(V,E)$ be a graph and let $ v \in V $. The \defin{open (resp. closed) neighborhood} of $ v $ in $ G $, denoted by $ N_G(v) \;( \text{resp. }N_G[v]) $, is defined as
$N_G(v) = \{ u \in V \mid \{u, v\} \in E \}$ and $N_G[v] = N_G(v)\cup\{v\}$.
If the graph $ G $ is clear from context, we may simply write $ N(v) $ and $N[v]$ instead of $N_G(v)$ and $N_G[v]$, respectively. Moreover, a vertex $v \in V$ is called \defin{unsaturated} if it is not contained in any triangles in $G$. 
The \defin{distance} between two vertices $ u $ and $ v $ in $ G $, denoted $ \mathsf{dist}_G(u, v) $, is the length (i.e., the number of edges) of a shortest path from $ u $ to $ v $ in $ G $. If no such path exists, we define $ \mathsf{dist}_G(u, v) = \infty $. When the graph is clear from context, we simply denote this distance by $ \mathsf{dist}(u, v) $.
One can also generalize the notion of distances to be between two sets of vertices.
Let $ A, B \subseteq V $. The \defin{distance} between the sets $ A $ and $ B $ in $ G $, denoted $ \mathsf{dist}_G(A, B) $, is defined as
$
\mathsf{dist}_G(A, B) = \mi \{ \mathsf{dist}_G(u, v) : u \in A,\, v \in B \}.
$

%\todo{Amirali:I suggest we do not use the word exploit. Exploit means using in a negative way. }
%\textcolor{blue}{Our useful structural property is summarized in the following lemma.} \todo{Need a better sentence}
\begin{lemma}\label{lem:lower_bound}
 Let $G$ be a graph such that every component of $G$ has at least three vertices. 
Then, there exists a minimum $\Delta$-completion set $F$ of $G$ such that for any edge $uv \in F$, we have $N_G(u) \cap N_G(v) \neq \emptyset$, that is, ${\mathsf{dist}}_G(u,v)=2$. 
%\todo{Amirali: we have not defined the neighbourhood notation}
\end{lemma}

\begin{proof}
Assume, for the sake of contradiction, that the statement does not hold. Let $  F $ be a minimum $\Delta$-completion set of $ G $ and let us color all edges of $ G $ blue and all edges in $ F $ red (dashed). Let $ \mathcal{C} = \{C_1, \ldots, C_k\} $ be a minimum set of triangles in $ H=G+F $ that triangle-covers all vertices of $ H $ and let $ b(\mathcal{C})$ be the number of triangles in $ \cal{C} $ with at least two red edges. Notice that each edge in $F$ must be an edge of some triangle in ${\cal C}$, since otherwise, $F$ is not minimum
%\todo{Anil: minimum?} \todo{Pawel: both are fine, indeed, minimum better}; 
consequently, if $b({\cal C})=0$,  the statement holds.
Our strategy is to construct a $ \Delta $-completion $ F^{\prime} $ of $ G $, with $ |F^{\prime}|=|F| $, possessing a minimum set $ \mathcal{C}^{\prime} $ of triangles in $ H^{\prime} = G+F^{\prime} $ that triangle-covers all vertices of $ H^{\prime} $ with 
$ b(\mathcal{C}^{\prime}) < b(\mathcal{C}) $, which, by applying the same argument iteratively, eventually results in a contradiction to our assumption.  

Since $ F $ is minimum, it follows that each triangle $ C_i \in \mathcal{C} $ has a vertex %, referred to as a~{\sl witness}, 
that is not covered by any other triangle in $ \mathcal{C} $. That is,  
$$
\forall C_i \in \mathcal{C}, \quad \exists v_i \in V(C_i) \text{ such that } v_i \notin V(C_j) \text{ for all } j \neq i.
$$
%\todo{Amirali: I don't understand why this is true. Also, what is a witness?}
The vertex $v_i$ is called a \defin{witness} for $C_i$.
Without loss of generality, assume that the triangle $ C_k $ has at least two red edges.  
Since $ C_k $ has a witness, there must exist a red edge $ e $ in $ C_k $ such that $ e $ does not belong to any other triangle $ C_i \in \mathcal{C} $ for $ i \in [k-1] $.  
We refer to such an edge $ e $ as an \defin{essential edge}.
We consider three cases:
\begin{itemize}[leftmargin=*]
\item  Case 1: $C_k$ has three essential edges. Assume $V\left(C_k\right)=\{x, y, z\}$. Consider $x_1 \in N_G(x)$ and $x_2 \in N_G\left(x_1\right) \sm\{x\}$, $y_1 \in N_G(y)$ and $y_2 \in N_G\left(y_1\right) \sm\{y\}$, and $z_1 \in N_G(z)$ and $z_2 \in N_G\left(z_1\right) \sm\{z\}$; see \cref{fig:fig_1} for an illustration. (Recall that $G$ has neither 1-vertex nor 2-vertex connected component, and so vertices $x_1, x_2, y_1, y_2, z_1, z_2$ exist; notice we may have $x_2=y$, etc.) Then, the triangle $C_k$ can be replaced with three triangles $C_k^x=\left(x, x_1, x_2\right), C_k^y=\left(y, y_1, y_2\right)$ and $C_k^z=\left(z, z_1, z_2\right)$ (or less if some of these triangles already belong to $\mathcal{C}$).

\begin{figure}[H]
    \centering
    \includegraphics[scale=0.5]{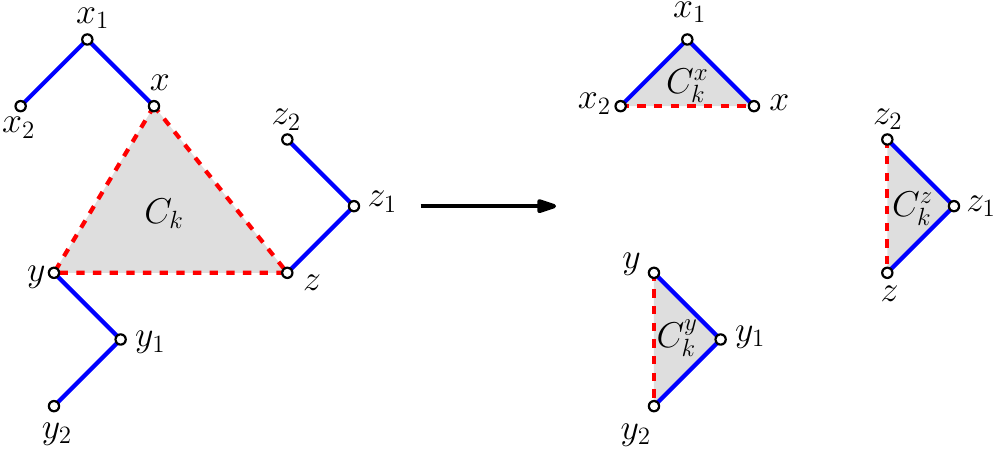}
    \caption{Case 1: $C_k$ has three essential edges.}
    \label{fig:fig_1}
\end{figure}

If $y=x_2$, $z=y_2$ and $x=z_2$, then our replacement keeps the same $\Delta$-completion set $F$, but results in another triangle-cover $\mathcal{C}^{\prime}$ satisfying $b(\mathcal{C}^{\prime}) < b( \mathcal{C})$. In any other case, our replacement results in a new minimum $\Delta$-completion set of $G$, the new graph $H^{\prime}$, and the new triangle-cover ${\cal C}^{\prime}$ with $b\left({\cal C}^{\prime}\right) < b({\cal C})$; notice that none of the edges $\{x,x_2\}, \{y,y_2\}$ and $\{z,z_2\}$ belongs to $F$ -- since otherwise, $F$ is not a minimum $\Delta$-completion set of $G$. 

\smallskip
\item Case 2: $C_k$ has two essential edges. We consider two subcases:
\begin{itemize}[leftmargin=*]
\item Subcase 2.a: $C_k$ has one blue edge. Assume $V\left(C_k\right)=\{x, y, z\}$ and without loss of generality, assume edge $\{x, z\}$ is blue. Consider $y_1 \in N_G(y)$ and $y_2 \in N_G\left(y_1\right) \sm\{y\}$, and $z_2 \in N_G(x) \sm\{z\}$; see \Cref{fig:fig_3} (a) for an illustration. (Recall that $G$ has neither 1-vertex nor 2-vertex connected component, and so vertices $y_1, y_2, z_2$ exist.) Then, the triangle $C_k$ can be replaced with two triangles $C_k^y=\left(y, y_1, y_2\right)$ and $C_k^z=\left(z, x, z_2\right)$ (or less if either $C_k^y \in \mathcal{C}$ or $C_k^z \in \mathcal{C}$). Our replacement results in a new minimum $\Delta$-completion set of $G$, the new  graph $H^{\prime}$, and the new triangle-cover ${\cal C}^{\prime}$ with $b\left({\cal C}^{\prime}\right) < b({\cal C})$. Notice that we can have neither $z=y_2$ nor $\{z,z_2\} \in F$ -- since otherwise, $F$ is not a  minimum $\Delta$-completion set of $G$ (recall edges $\{x,y\}$ and $\{y,z\}$ are not essential). 

\begin{figure}[H]
\centering \includegraphics[scale=0.4]{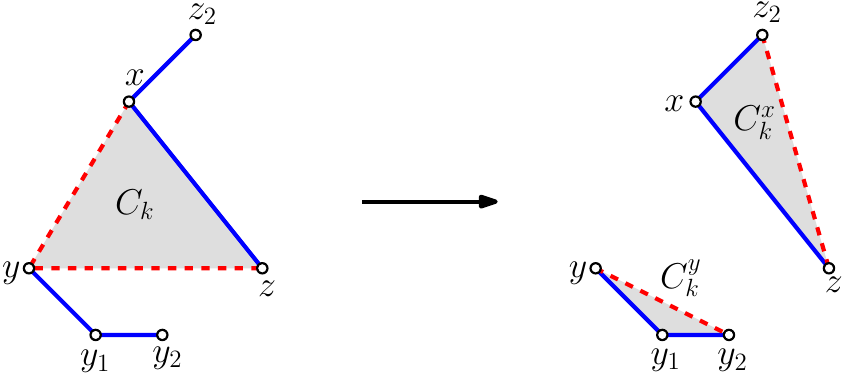}
\phantom{XXXXXX}\includegraphics[scale=0.4]{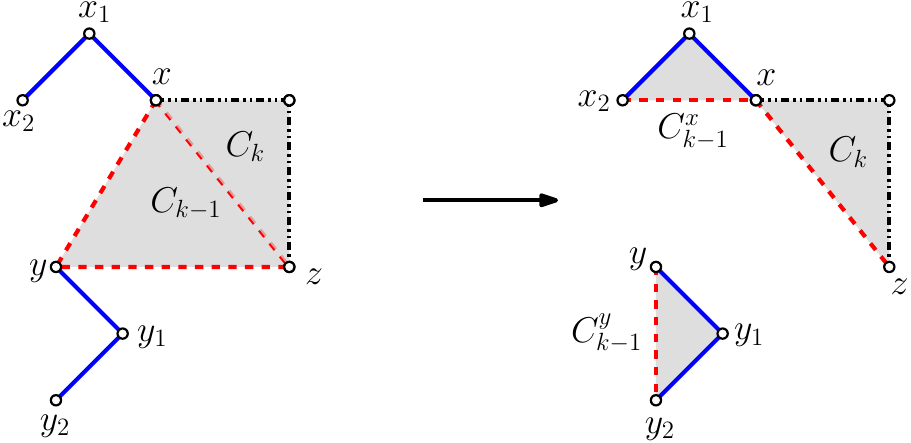}
\caption{(a) Case 2.a: $C_k$ has only one blue edge. (b) Case 2.b: $C_k$ has no blue edge}
\label{fig:fig_3}
\end{figure}
%\todo{In Caption 2b, we can have $C_k$ and $C_{k-1}$ to match the description in text.}
\item Subcase 2.b: $C_k$ has no blue edge. 
Assume $V\left(C_k\right)=\{x, y, z\}$ and without loss of generality assume edge $\{x, z\}$ is not essential, say $\{x, z\} \in E\left(C_{k-1}\right)$. Consider $x_1 \in N_G(x)$ and $x_2 \in N_G\left(x_1\right) \sm\{x\}$, and $y_1 \in N_G(y)$ and $y_2 \in N_G\left(y_1\right) \sm\{y\}$; see \cref{fig:fig_3} (b) for an illustration. (Recall that $G$ has neither 1-vertex nor 2-vertex connected component, and so vertices $x_1, x_2, y_1, y_2$ exist;   notice we may have $x_2=y$, etc.) Then, the triangle $C_k$ can be replaced with two triangles $C_k^x=\left(x, x_1, x_2\right)$ and $C_k^y=\left(y, y_1, y_2\right)$ (or less if some of these triangle already belongs to $\mathcal{C}$ ).

If $y=x_2$ and $z=y_2$, then our replacement keeps the same $\Delta$-completion set $F$, but results in another triangle-cover $\mathcal{C}^{\prime}$ satisfying $b(\mathcal{C}^{\prime}) < b( \mathcal{C})$. In any other case, our replacement results in a new minimum $\Delta$-completion set of $G$, the new graph $H^{\prime}$, and the new triangle-cover ${\cal C}^{\prime}$ with $b\left({\cal C}^{\prime}\right) < b({\cal C})$; notice that neither $\{x,x_2\}$ nor $\{z,z_2\}$ belongs to $F$ -- since otherwise, $F$ is not a minimum $\Delta$-completion set of $G$. 
\item Case 3: $C_k$ has one essential edge. We have two subcases: $C_k$ has one blue edge or none, but we can handle them together (see \cref{fig:fig_5} (a) and \cref{fig:fig_5} (b) for illustrations). 
Assume $V\left(C_k\right)=\{x, y, z\}$ and without loss of generality assume edge $\{x, y\}$ is essential while $\{x, z\}$ is blue. Consider $x_1 \in N_G(x)$ and $x_2 \in N_G\left(x_1\right) \sm\{x\}$. (Recall that $G$ has neither 1 -vertex nor 2-vertex connected component, and so vertices $y_1, y_2$ exist.) Then, the triangle $C_k$ can be replaced with the triangle $C_k^x=\left(x, x_1, x_2\right)$ (notice $C_k^x \notin \mathcal{C}$).
    
If $y=x_2$ then our replacement keeps the same $\Delta$-completion set $F$, but results in another triangle-cover $\mathcal{C}^{\prime}$ satisfying $b(\mathcal{C}^{\prime}) < b( \mathcal{C})$. Otherwise, our replacement results in a new minimum $\Delta$-completion set of $G$, the new  graph $H^{\prime}$, and the new triangle-cover ${\cal C}^{\prime}$ with $b\left({\cal C}^{\prime}\right) < b({\cal C})$; notice that $\{x,x_2\} \notin F$ -- since otherwise, $F$ is not a  minimum $\Delta$-completion set of $G$ (recall $\{x,y\}$ is not essential).

\begin{figure}[H]
        \centering
        \includegraphics[scale=0.4]{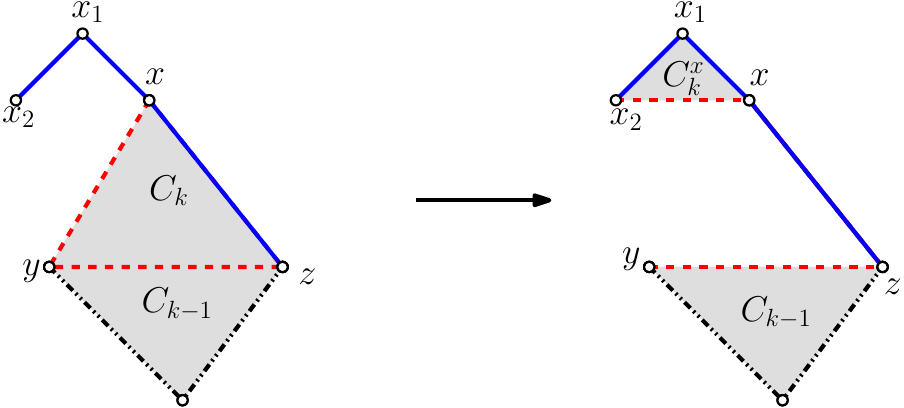}
%        \caption{Case 3: $C_k$ has one blue edge.}
        %\label{fig:fig_4}
%    \end{figure}
%\begin{figure}[H]
%    \centering
\phantom{XXXXX}\includegraphics[scale=0.4]{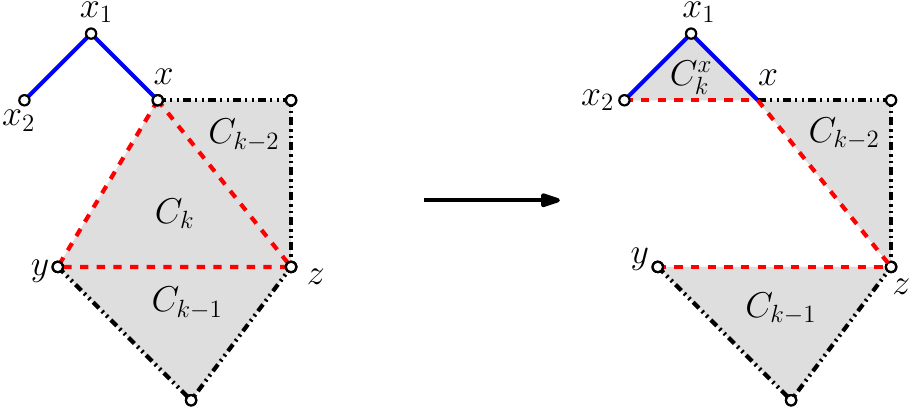}
    \caption{Case 3: (a) $C_k$ has one blue edge; (b) $C_k$ has no blue edge}
    \label{fig:fig_5}
\end{figure}
\end{itemize}
\end{itemize}

\noindent Applying the same argument iteratively, eventually results in a contradiction to our assumption.  
\end{proof}
Observe that in the proof of~\Cref{lem:lower_bound}, we did not take any advantage of the fact that all vertices of a graph must be $\Delta$-completed. Therefore, we can generalize the notion of $\Delta$-completion set of a graph $ G $ into $\Delta$-completion set of a subset of vertices of $ G $.
More precisely, an edge set $ F \subseteq \binom{V(G)}{2} \sm E(G)$ is a $\Delta$-completion set of $ X $ if every vertex in $ X $ lies in a triangle in $ G + F $. The proof of~\Cref{lem:lower_bound} leads immediately to the following corollary.
\begin{corollary}\label{cor:lower_bound_for_subset}
 Let $G$ be a graph such that every component of $G$ has at least three vertices, and let $X$ be a set of vertices of $G$.
Then, there exists a minimum $\Delta$-completion set $F$ of $X$ such that for any edge $uv \in F$, we have $N_G(u) \cap N_G(v) \neq \emptyset$, that is, ${\mathsf{ dist}}_G(u,v)=2$. 
\end{corollary}
It follows from \cref{lem:lower_bound} (resp.~\Cref{cor:lower_bound_for_subset}) that when constructing a minimum set of edges whose addition (if necessary) ensures that each vertex is covered by a triangle, we may always restrict our attention to either the original triangles or those formed by adding edges whose endpoints are at distance two apart in the original graph $ G $.  
In other words, we only need to consider triangles that contain at least two original edges. As a direct consequence, we obtain the following corollary.

\begin{corollary}\label{dist=3}
Let $G$ be a graph such that each component has at least $3$ vertices, and let $X_1, \ldots X_t$ be subsets of vertices of $G$ such that:
\begin{itemize}
\item For any distinct $i, j \in[t]$, we have ${\mathsf{ dist}}_G\left(X_i, X_j\right) \geq 3$;
\item In $G$, at least $a_i$ edges are needed to be added in order to $\Delta$-complete vertices in $X_i$.
\end{itemize}
%\todo{Amirali: how is the distance function over two sets of vertices defined?}
Then, we need to add at least $\sum_{i=1}^t a_i$ edges to $G$ in order to $\Delta$-complete all vertices of $G$.
\end{corollary}

\section{Complexity}
In this section, we first prove that determining whether a graph $G$ has a $\Delta$-completion set of at most $t$ edges is $\np$-complete. From the hardness reduction, we can also conclude that it is hard to approximate $\Delta_G$ within any constant factor. Next, we provide a $(\ln n+1)$-factor approximation algorithm.

\subsection{\texorpdfstring{$\np$}{NP}-Hardness for General Graphs}

 We prove the hardness of the triangle-covered problem by a reduction from the $\np$-complete problem of SET-COVER. An instance $(\mathbb X, \mathbb F, t)$ of SET-COVER consists of a set of items $\mathbb X$, a family $\mathbb F$ of the subsets of $\mathbb X$ such that no set in $\mathbb F$ is empty and each item in $\mathbb X$ belongs to at least one set in $\mathbb F$, and an integer $t\geq 0$. An instance $(\mathbb X, \mathbb F, t)$ is a YES-instance of SET-COVER if and only if there exists a collection of subsets $\mathbb S\subseteq \mathbb F$ such that $\bigcup_{S\in \mathbb S}S=\mathbb X$ and $|\mathbb S|\leq t$.

The following result will help us prove the inapproximability of the triangle-covered problem.

\begin{figure}[H]
\centering
\subfloat[\centering An example of the reduction graph $\greductiontwo$ for $\mathbb{F}=\{S_1, S_2,S_3\}$, $\mathbb{X}=\{x_1,x_2,x_3\}$, $S_1=\{x_1\}$, $S_2=\{x_2\}$, and $S_3=\{x_2,x_3\}$. The unsaturated vertices are depicted in red. ]{{\includegraphics[scale=0.4]{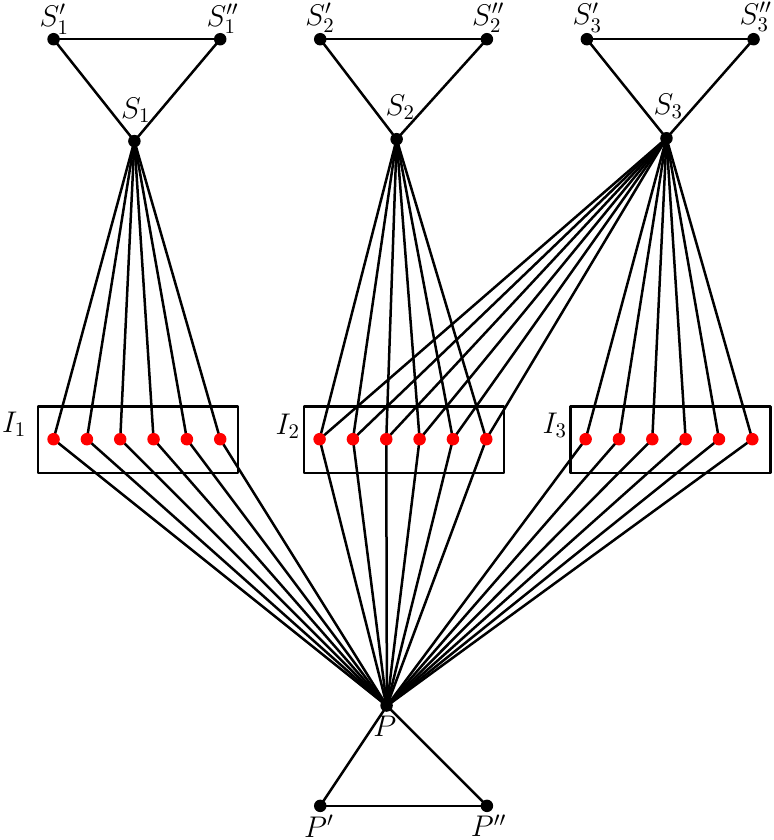}}}
\qquad
\subfloat[\centering An example of the bipartite reduction graph $\greductiontwo$ for $\mathbb{F}=\{S_1, S_2,S_3\}$, $\mathbb{X}=\{x_1,x_2,x_3\}$, $S_1=\{x_1\}$, $S_2=\{x_2\}$, and $S_3=\{x_2,x_3\}$. The unsaturated vertices are depicted in red.]{{\includegraphics[scale=0.4]{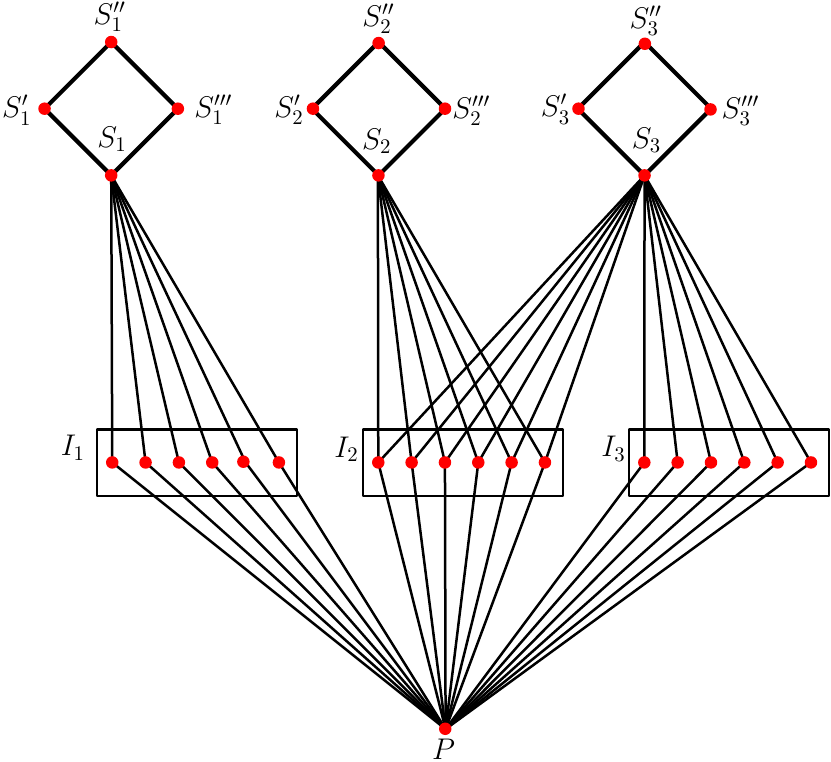} }}%
\caption{}
\label{figreductionk3}
\end{figure}

\begin{lemma}{\rm \cite[Corollary 4]{dinur2014analytical}}
For every $\varepsilon >0$, there exists no polynomial-time $((1-\varepsilon) \cdot \ln |\mathbb{X}|)$- approximation algorithm for SET-COVER, unless $\p=\np$. 
\label{setcoverhard}
\end{lemma}
Given an instance $(\mathbb X, \mathbb F, t)$ of SET-COVER, we construct a graph $\greductiontwo=(\vreductiontwo,\ereductiontwo)$ as follows (see \Cref{figreductionk3} (a) for an illustration). 
\begin{enumerate}
\item Initially, $\vreductiontwo=\emptyset$ and $\ereductiontwo=\emptyset$.
    \item For every set $S_j \in \mathbb F$, we add a vertex $S_j$ to $\vreductiontwo$. We refer to such vertices as \emph{set vertices}.
    \item For every item $x_i \in \mathbb X$, we add a subgraph $I_i$ to $\greductiontwo$. Each $I_i$ is a disjoint union of $2\numitems$ isolated vertices. We refer to each $I_i$ as an \defin{item subgraph}. We update $\vreductiontwo$ accordingly.
    \item For each $x_i \in \mathbb X$ and every $S_j \in \mathcal{F}$, if $x_i \in S_j$, we connect $S_j \in \vreductiontwo $ to all $2\numitems$ vertices of $I_i$, i.e., $\forall v \in V(I_i):\ereductiontwo \xleftarrow{} \ereductiontwo \cup \{S_jv\}$.

    \item We add a vertex $P$ to $\greductiontwo$, $\vreductiontwo \xleftarrow{} \vreductiontwo \cup \{P\}$. For every item subgraph $I_i$, we connect every vertex of $I_i$ to this new vertex, i.e., $\forall I_i \forall v\in V(I_i): \ereductiontwo \xleftarrow[]{} \ereductiontwo \cup \{vP\}$. We refer to vertex $P$ as \defin{the pivot vertex}.
        \item For every vertex $S_j$ added in Step 2, we cover it in a triangle by adding two vertices $S'_j$ and $S''_j$ to $\vreductiontwo$ by setting $\vreductiontwo \xleftarrow[]{}\vreductiontwo \cup \{S'_j, S''_j\}$ and $\ereductiontwo \xleftarrow{}\ereductiontwo \cup \{S_jS'_j,S_j S''_j, S'_j S''_j\}$. Similarly, we cover the common vertex $P$ (Step 5) in a triangle by adding two new vertices $P'$ and $P''$ and setting $\vreductiontwo \xleftarrow[]{}\vreductiontwo \cup \{P', P''\}$ and $\ereductiontwo \xleftarrow{}\ereductiontwo \cup \{P P',P P'', P' P''\}$. At the end of this step, all edges of $\greductiontwo$ are contained in triangles. We refer to these new vertices $P', P'', S'_j, S''_j$ as \defin{auxiliary vertices}.
\end{enumerate}

Through the following observation, we show that for any instance $(\mathbb X, \mathbb F, t)$ of SET-COVER, the size of the constructed graph $\greductiontwo$ is polynomial in $\numitems + \numsets$.
\begin{observation}\label{sizeofreductiongraph}
    For any instance $(\mathbb X, \mathbb F, t)$ of SET-COVER, let $\greductiontwo$ be the constructed reduction graph. Then, $|\vreductiontwo|\leq 2\numitems^2 +3\numsets+3$.
\end{observation}

For simplicity, in the remainder of this section, $(\mathbb X, \mathbb F, t)$ and $\greductiontwo$ serve as an arbitrary instance of SET-COVER and its corresponding reduction graph, respectively. The following definition will help us prove our hardness result. 
\begin{definition}\label{dfnproper}
Let $E'$ be a $\Delta$-completion set of $\greductiontwo$. We say $E'$ is a \defin{proper $\Delta$-completion set} of $\greductiontwo$ if $E' \subseteq \{S_1P,\dots,S_{\numsets}P\}$.
\end{definition}
\noindent We have the following lemma. 
\begin{lemma}\label{obspropertosetcover}
    Any proper $\Delta$-completion set $E'$ of $\greductiontwo$ corresponds to a set cover of size $|E'|$  for  $(\mathbb X, \mathbb F, t)$.
\end{lemma}
\begin{proof}
    Let $E'$ be any proper $\Delta$-completion set of $\greductiontwo$. For any edge $S_jP \in E'$, we pick the set $S_j$ and add it to an initially empty collection of sets $\mathbb S$. For any item $x_i \in \mathbb X$, there must be an edge $S_{\ell}P \in E'$ with $x_i \in S_{\ell}$ (since $E'$ is a proper $\Delta$-completion set of $\greductiontwo$). Therefore, item $x_i$ is covered by $S_{\ell}$ in $\mathbb S$. Since this analysis applies to any item $x_i$, it follows that $\mathbb S$ is a set cover of size $|E'|$ for $(\mathbb X, \mathbb F, t)$.  \end{proof}

To prove the $\np$-completeness of the triangle-covered problem, we use the property of proper $\Delta$-completion sets mentioned in \Cref{obspropertosetcover}. The main idea behind the forthcoming $\np$-hardness proof is that the reduction graph $\greductiontwo$ always has an optimal $\Delta$-completion set which is also proper. The proof of the existence of such an optimal $\Delta$-completion set is constructive, we show in \cref{lemconstruction} that \Cref{nphardconstruction2} converts any $\Delta$-completion set $E'$ of $\greductiontwo$ into  a proper $\Delta$-completion set $E''$ of $\greductiontwo$ with $|E''|\leq |E'|$ in time polynomial in $\numsets +\numitems$. 
\begin{algorithm}[H]
    \caption{}
    \label{nphardconstruction2}
    \begin{algorithmic}[1]
            \State \textbf{Input:} $\greductiontwo$, $E'$ (a $\Delta$-completion set of $\greductiontwo$) 
            \State \textbf{Output:} $E''$, a proper $\Delta$-completion set of $\greductiontwo$ with $|E''|\leq |E'|$
            \State \textbf{Initialization:} $E''\xleftarrow{}\emptyset$
            \State \textbf{Step 1:} For every pair of sets $S_j$ and $S_{\ell}$ (with $j<\ell$ and $\{S_j,S_{\ell}\}\subseteq \mathbb F$), if $S_j S_{\ell}\in E'$, then set $E''\xleftarrow{}E'' \cup \{S_jP\}$.
            \State \textbf{Step 2:} For every item $x_i \in \mathbb X$, let $S_j \in \mathbb F$ be a set with $x_i \in S_j$. If $E'$ has an edge $uv$ such that $\{u,v\}\subseteq V(I_i)$ for its corresponding item subgraph $I_i$, then add $S_jP$ to $E''$, i.e., $E''\xleftarrow{}E'' \cup \{S_jP\}$.
            \State \textbf{Step 3:} Add the \emph{proper} subset of $E'$ to $E''$, i.e., $E'' \xleftarrow{} E'' \cup (E' \cap \{S_jP|\; S_j \in \mathcal F\})$.
            \State \textbf{Step 4:} For each item $x_i \in \mathbb X$, if $I_i$ has a vertex $v\in V(I_i)$ such that $v$ is unsaturated in $\greductiontwo \cup E'' $, then add the edge $S_jP$ to $E''$ for some set $S_j \in \mathbb{F}$ with $x_i \in S_j$.
            \State \Return $E''$. 
    \end{algorithmic}
\end{algorithm}
\begin{lemma}\label{lemconstruction}
Let $E'$ be a $\Delta$-completion set of $\greductiontwo$. \Cref{nphardconstruction2} returns a proper $\Delta$-completion set $E''$ of $\greductiontwo$ with $|E''|\leq |E'|$ in time polynomial in $\numsets +\numitems$.     
\end{lemma}
   \begin{proof}
    That $E''$ is a proper $\Delta$-completion set follows from the fact that we explicitly only add edges of the desired type $S_j P$ (\Cref{dfnproper}) to $E''$. Moreover, it can be seen that $\Cref{nphardconstruction2}$ terminates in time polynomial in $\numsets+\numitems$.

    We now show $|E''|\leq |E'|$. We show this bound by finding a unique counterpart for each edge added to $E''$ in a way such that each edge in $E'$ has at most one counterpart in $E''$. For each edge $S_j P$ added to $E''$ in the first step, we designate the edge $S_j S_{\ell}\in E'$ (with $j < \ell$) as its unique counterpart (see Line 4). For any edge $S_j P$ added to $E''$ in Step 2, we designate $uv \in E'$ (with $\{u,v\}\subseteq V(I_i)$ and $x_i \in \mathbb X$) as its counterpart (see Line 5). Moreover, the unique counterpart for each edge $S_jP\in E''$ added in Step 3 is its copy $S_j P\in E'$ (Line 6).

     We now show that if there exists a vertex $v$ as described in Line 7 of \Cref{nphardconstruction2} at the beginning of Step 4, then there must exist at least $2\numitems$ edges in $E'$, none of which has been designated as the counterpart of any edge in $E''$. First, observe that since $v$ is unsaturated in $\greductiontwo\cup E''$, there cannot be an edge $S_j S_{\ell}\in E'$ with $x_i \in S_j \cap S_{\ell}$ in the corresponding SET-COVER instance. If such an edge existed, then due to Step 1 of \Cref{nphardconstruction2} we would have $S_jP\in E''$ or $S_{\ell} P\in E''$, implying that $v$ is not unsaturated in $G\cup E''$. By similar reasoning and observing Step 2 and Step 3, there exist no edges $(u,v) \in E'$ (with $\{u,v\}\subseteq V(I_i)$) and $S_jP \in E'$ (with $x_i \in S_j$ in the SET-COVER instance). It follows that each vertex $u \in V(I_i)$ is saturated in $E'$ by an edge incident on $u$. Since in Steps 1 to 3 we never designate any edge of $E'$ incident on a vertex of an item subgraph, we deduce that there exist $2|\mathbb X|$ edges in $E'$ that do not have a counterpart in $E''$. Since Step 4 adds at most $|\mathbb X|$ edges to $E''$, there will be enough designated counterparts in $E'$.
\end{proof}
%\begin{lemma}{\rm \cite[Proposition 1.1.]{chakraborti2020extremal}}
%For any positive integers $q, n, t$ with $2 \leq t \leq n$, and any integer
%$N = qn+r$ with $0 \leq r < n$, the graph consisting of the union of 2 copies of $K_n$ sharing
%$n - r$ vertices, together with the disjoint union of $q - 1$ many $K_n$, has the least number
%of copies of $K_t$ among all $K_n$-covered graphs on N vertices. Moreover, this is the
%unique such graph if $n - r \geq t$ or $n - r = 1$.
%\end{lemma}

%\noindent\textbf{\color{red}PZ:} {\color{blue} Moved above the lemma below}. 

%\begin{lemma}{\rm \cite[Corollary 4]{dinur2014analytical}}
%For every $\varepsilon >0$, there exists no polynomial-time $((1-\varepsilon). \ln |\mathcal{X}|)$- approximation algorithm for SET-COVER, unless $\p=\np$. 
   % \label{setcoverhard}
%\end{lemma}
%\begin{theorem}
%\label{main_3}
%The triangle-covered problem is $\mathbb{NP}$-complete. Moreover, there exists no polynomial-time $c$-approximation algorithm for the triangle-covered problem for any constant $c>1$, unless $\p=\np$. 
%\end{theorem}
\mainthree*
\begin{proof}
Given a graph $G=(V,E)$ and a set of edges $E'\subseteq (V \times V) \sm E$, it can be verified in polynomial time whether $E'$ is a $\Delta$-completion set of $G$ or not. Therefore, the triangle-covered problem is in $\np$.
To show the $\np$-hardness, we show that an instance $(\mathbb X, \mathbb F, t)$ is a YES-instance of SET-COVER if and only if its corresponding reduction graph $\greductiontwo=(\vreductiontwo, \ereductiontwo)$ has a $\Delta$-completion set of size at most $t$. Indeed, if $(\mathbb X, \mathbb F, t)$ has a set cover of size at most $t$, we can construct a $\Delta$-completion set of the same size consisting of edges $S_jP$ for each set $S_j$ in this set cover. Conversely, if $\greductiontwo$ has a $\Delta$-completion set $E'$ of size at most $t$, then by \Cref{lemconstruction} it has a proper $\Delta$-completion set of size at most $t$, which corresponds to a set cover of the same size for $(\mathbb X, \mathbb F, t)$. Therefore, the triangle-covered problem is $\np$-hard.

To prove the hardness of approximation, suppose to the contrary that $\p\neq \np$ and there exists a polynomial-time $c$-approximation algorithm $\alg$ for the triangle-covered problem for a constant $c>1$. We show how to use $\alg$ to devise a constant-factor approximation algorithm for SET-COVER. 
For any instance $(\mathbb X, \mathbb F, t)$ of SET-COVER, construct the graph $\greductiontwo=(\vreductiontwo, \ereductiontwo)$ in time polynomial in $\numsets+\numitems$ (\Cref{sizeofreductiongraph}). By \Cref{lemconstruction}, it is easy to see that the size of the minimum $\Delta$-completion set of $\greductiontwo$ is equal to the size of the optimal set cover for $(\mathbb X, \mathbb F, t)$. By applying $\alg$ to $\greductiontwo$, we get a $\Delta$-completion set $E'$ whose size is at most $c$ times bigger than the minimum $\Delta$-completion set.  If $E'$ is not a proper $\Delta$-completion set, then by \Cref{lemconstruction} it can be transformed into a proper $\Delta$-completion set $E''$ with $|E''|\leq |E'|$ in polynomial time. However, the resulting proper $\Delta$-completion set corresponds to a set cover for $(\mathbb X, \mathbb F, t)$ whose size is within a factor of c of the optimal set cover. It follows that the described procedure is a polynomial-time $c$-approximation algorithm for the SET-COVER problem. By \Cref{setcoverhard}, we have $\p=\np$, contradicting our original assumption that $\p\neq \np$.
\end{proof}
With the help of \Cref{lem:lower_bound}, we can modify the reduction described in the previous section and show that the triangle-covered problem remains $\np$-complete, even for connected bipartite graphs. See \Cref{figreductionk3} (b). 
%We defer the proof to the appendix.
\begin{theorem}
    The triangle-covered problem is $\np$-complete, even when its input graph is connected and bipartite. 
\end{theorem}

\begin{proof}
    We have already shown that the triangle-covered problem is in $\np$. To show $\np$-hardness, we use the same reduction graph $\greductiontwo$ with a few adjustments to make it bipartite (see \Cref{figreductionk3} (b) for an illustration). 

    Given an instance $(\mathbb X, \mathbb F, t)$ of SET-COVER, we construct the graph $\greductiontwo=(\vreductiontwo,\ereductiontwo)$ in the same way as before but with a different last step. In the last step, for each set vertex $S_j$, we cover it in a cycle of length four by adding three new auxiliary vertices $S'_{j}, S''_{j}, S'''_j$ such that $S''_j$ is at distance two from $S_j$ in $\greductiontwo$ for all $j$ (see \Cref{figreductionk3} (b)). We refer to these copies of $C_4$ as \emph{set subgraphs}. For the $j$-th set subgraph, we call non-edges $S_{j}S''_j$ and $S'_jS'''_j$ the \emph{diagonals} of that set subgraph. Moreover, we do not add any auxiliary vertices for the pivot vertex. $\greductiontwo$ has no odd cycles and is therefore bipartite. Finally, the connectivity of $\greductiontwo$ is due to the existence of the pivot vertex $P$.

\begin{claim}\label{clmnphardbipartite}
    $(\mathbb X, \mathbb F, t)$ has a set cover of size at most $t$ if and only if $\greductiontwo$ has a $\Delta$-completion set of size at most $t+\numsets$.
\end{claim}
\begin{clmproof}
Suppose $(\mathbb X, \mathbb F, t)$ has a set cover of size at most $t$. We construct a $\Delta$-completion set $E'$ of size at most $\numsets +t$ as follows. For each set $S_j$ in this set cover, we add the edge $S_jP$ to $E'$ in the first step. After this step, the vertices of all item subgraphs are covered in triangles in $G\cup E'$ with $|E'|\leq t$. Furthermore, since $(\mathbb X, \mathbb F, t)$ has at least one item and one set (by assumption), $P$ is also covered in a triangle in $G \cup E'$. In the second step, we add to $E'$ exactly $\numsets$ edges consisting of the diagonal $S_jS''_j$ for the $j$-th set subgraph, $j\in \{1,\dots, \numsets\}$. After this step, all vertices of $\greductiontwo \cup E'$ are covered in triangles with $|E'|\leq \numsets +t$.
    Now suppose $\greductiontwo$ has a $\Delta$-completion set of size at most $\numsets +t$. By \Cref{lem:lower_bound}, $\greductiontwo$ must have a minimum $\Delta$-completion set $E^*$ in which the endpoints of each edge are at distance two in $\greductiontwo$ and $|E^*|\leq \numsets+t$. This property of $E^*$ implies that vertex $S''_j$ in the $j$-th set subgraph must be covered in $G\cup E^*$ by a diagonal $S'_jS'''_j$ or $S_jS''_j$ of that set subgraph. Therefore, $E^*$ must contain a diagonal of each set subgraph. Let $D\subseteq E^*$ denote the set of these diagonals (with $|D|=\numsets$), and let $E'\xleftarrow{}E^* \sm D $. Observe that $$|E'|\leq \numsets+t-\numsets =t.$$
    Apply \Cref{nphardconstruction2} to $E'$ to obtain another set $|E''|\leq |E'|\leq t$. In the proof of \Cref{lemconstruction}, observe that we never designate any edge in $D$ as a counterpart of an edge in $E''$. Therefore, the vertices of all item subgraphs, as well as the pivot vertex $P$, are covered in triangles in $G\cup E''$. For each edge $S_jP \in E''$, pick the set $S_j\in \mathbb F$ and add it to an initially empty collection of sets $\mathcal C$. It is easy to see that $\mathcal C$ covers all items in $\mathbb X$ with $|\mathcal{C}|=|E''|\leq t$, implying that $(\mathbb X, \mathbb F, t)$ has a set cover of size at most $t$.
\end{clmproof}
The $\np$-hardness proof follows from \Cref{clmnphardbipartite}. 
\end{proof}

\subsection{An Approximation Algorithm for General Graphs}
So far, we have seen that the triangle-covered problem admits no polynomial-time $c$-approximation algorithm for any constant $c > 1$, unless $\mathbb P =\mathbb {NP}$ (\Cref{main_3}). In this section, we provide a greedy $(\ln n+1)$-approximation algorithm for general graphs.

Any instance of the triangle-covered problem can be formulated as an instance of the SET-COVER problem as follows.
Let $ U $ be the set of unsaturated vertices of a graph $G=(V,E)$.
For each non-edge $ uv \notin E $ with $\mathsf{dist}(u, v)=2$, define a subset $S_{uv} = \{ w \in U \mid u w \in E \text{ and } v w \in E \}\cup \{\{u,v\}\cap U\}$. For such a non-edge $uv$, $S_{uv}$ contains the set of all unsaturated vertices that lie on a path of length two between $u$ and $v$ (possibly including $u$ and $v$). Construct an instance $\mathcal I=(\mathbb X, \mathbb F) $ of SET-COVER, where $\mathbb X=U$ and $\mathbb F=\bigcup_{uv\in E'} S_{uv}$ with $E'=\{uv|uv \text{ is a non-edge with }\mathsf{dist}(u, v)=2\}$. Since adding any non-edge $uv $ of $E'$ to $G$ saturates all $ w \in S_{uv} $, any set cover for $\mathcal I$ of size $t$ corresponds to a $\Delta$-completion set of size $t$ for $G$. Let $\opt_{\mathcal I}$ denote the size of the minimum set cover for $\mathcal I$. 
 \begin{claim}\label{opteq}
     $\Delta_G=\opt_{\mathcal I}$.
 \end{claim}
 \begin{proof}
     Since any set cover for $\mathcal I$ of size $t$ corresponds to a $\Delta$-completion set of size $t$ for $G$, we have $\Delta_G \leq \opt_{\mathcal{I}} $. We now show $\opt_{\mathcal{I}} \leq \Delta_G  $. For the sake of contradiction, suppose $\opt_{\mathcal{I}} > \Delta_G  $. By \Cref{lem:lower_bound}, there exists a minimum $\Delta$-completion set $F$ of $G$ with $|F|=\Delta_G$ such that $F \subseteq E'$. For every non-edge $uv\in F$, we set $S\xleftarrow{} S\cup S_{uv}$ (initially, $S\xleftarrow{}\emptyset$). It can be seen that $S$ is a set cover for $\mathcal I$ with $|S|=|F|=\Delta_G<\opt_{\mathcal{I}}$, a contradiction. 
 \end{proof}
It is well known that the greedy algorithm for SET-COVER outputs the cover $S$, where $|S|\leq (\ln | \mathbb X|+1).\opt_{\mathcal I}$.
% 
% There exists a polynomial-time $(\ln |\mathbb{X}|+1)$-approximation algorithm for SET-COVER. We denote this algorithm by $\alg$ and omit its details for brevity (see~\cite{Chvatal79} for details). Apply $\alg$ to $\mathcal{I}$ to obtain a set cover $S$ of $\mathcal I$ with $|S|\leq (\ln | \mathbb X|+1).\opt_{\mathcal I}$. 
Since $S$ corresponds to a $\Delta$-completion set $F$ of $G$, we have  
 $$|F|=|S|\leq (\ln | \mathbb X|+1).\opt_{\mathcal I}=(\ln | \mathbb X|+1)\cdot\Delta_G=(\ln | U|+1)\cdot\Delta_G\leq (\ln n+1)\cdot\Delta_G$$
 which implies the following theorem.
\begin{theorem}
Let $G$ be a graph of order $n$ such that every component of $G$ has at least three vertices. There exists a polynomial-time $(\ln n+1)$-approximation algorithm for the triangle-covered problem.
\end{theorem}

\section{Bounds and Algorithms for \texorpdfstring{$\Delta$}{Delta}-Completion Sets}
This section presents algorithms for constructing minimum $\Delta$-completion sets for trees and chordal graphs. We conclude this section by providing bounds for cactus graphs.

\subsection{Trees}
We now present an efficient algorithm for finding a minimum 
$\Delta$-completion set in trees.
To achieve this, we first provide an optimal algorithm for stars, and then use this algorithm as a subroutine to develop an efficient algorithm for trees with diameter at least 3.
In addition, we show that the upper bound on the size of the minimum $ \Delta $-completion set is $ \frac{n}{2} $ for any tree of order $n$.
Recall that a vertex of degree $1$ is a leaf.

\begin{prop}\label{prop:star}
Let $S_n$ be a star of order $n\ge 3$. Then $\Delta_{S_n}=\lfloor\frac{n}{2}\rfloor$,
\end{prop}

\begin{proof}
Let $F$ be a minimum $\Delta$-completion set of $S_n$.
There are $ n-1 $ leaves to cover and each edge in $ F $ covers at most 2 leaves. 
Therefore, the lower bound is $\lceil \frac{n-1}{2}\rceil$ which is equal to $\lfloor \frac{n}{2}\rfloor$.
\end{proof}

Next, we establish an upper bound for the minimum $\Delta$-completion set of a double star.
A \emph{double star} is a tree formed by taking an edge (called the central edge) and attaching some number of leaves to each of its two endpoints.
More formally, let $ u $ and $ v $ be two vertices connected by an edge. Attach $ a \geq 1 $ leaves to $ u $, and $ b \geq 1 $ leaves to $ v $. The resulting tree is called a \defin{double star}, and is denoted by $ S_{a,b} $.

\begin{prop}\label{prop:double_star_exa}
Let $T=S_{\ell_1,\ell_2}$ be a double star of order $n$ with centers $c_i$, where $c_i$ has $\ell_i$ leaves for $i=1,2$. Then $$
\Delta_T = 
\begin{cases}
\lceil\frac{n}{2}\rceil, & \text{if $\ell_1,\ell_2$ are odd}, \\
\lceil\frac{n}{2}\rceil -1, & \text{otherwise}
\end{cases}
$$
\end{prop}
\begin{proof}
Let $T$ be a double star with centers $c_i$, where $c_i$ has $\ell_i$ leaves for $i=1,2$.
We first find a set of paths of length $ 2 $, that cover each vertex (see \cref{fig:tri_2}).  
For each such path, we add one edge. This gives us the following upper bound:
$$
\left\lceil\frac{\ell_1}{2}\right\rceil+\left\lceil\frac{\ell_2}{2} \right\rceil = 
\begin{cases}
\lceil\frac{n}{2}\rceil, & \text{if $\ell_1,\ell_2$ are odd}, \\
\lceil\frac{n}{2}\rceil -1, & \text{otherwise}
\end{cases}
$$
Now, let $ F $ be a minimum $ \Delta $-completion set obtained by \Cref{lem:lower_bound}.  
This implies that the endpoints of every edge in $ F $ have a common neighbour in $ T $.  
If $ uc_i \in F $, where $ u $ is a leaf adjacent to $ c_j $, with $ i \neq j \in \{1,2\} $, then we replace $ uc_i $ with $ uv $, where $ v $ is a leaf adjacent to $ c_j $.
\begin{figure}[H]
    \centering
    \subfloat[\centering ]{{\includegraphics[scale=0.48]{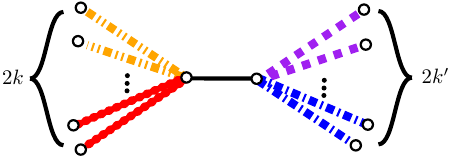}}}
\qquad
\subfloat[\centering ]{{\includegraphics[scale=0.48]{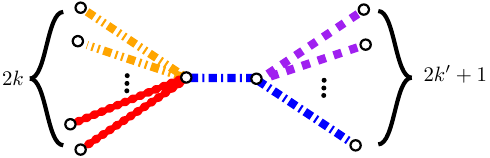} }}%
\qquad
\subfloat[\centering ]{{\includegraphics[scale=0.48]{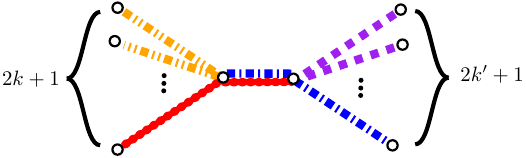} }}%
\caption{Double stars with $\Delta$-completion set.}
\label{fig:tri_2}
\end{figure}
It is worth mentioning that the leaves are paired so that we get half the number of leaves. 
This establishes a matching lower bound, hence the bound is tight.\end{proof}
Next, we generalize the previous result to an arbitrary tree.

\begin{theorem}\label{thm:tree}
Let $n \ge 3$. 
Then the following holds:
\[
   \min_{T\text{ tree of order }n} \Delta_T \;=\;
   \Bigl\lceil \tfrac{n}{3}\Bigr\rceil,
   \qquad
   \max_{T\text{ tree of order }n} \Delta_T \;=\;
   \Bigl\lceil \tfrac{n}{2}\Bigr\rceil.
\]
\end{theorem}

\begin{proof}
Let $T$ be a tree of $n \ge 3$ vertices. We construct a $\Delta$-completion set of size at most $\lceil\frac{n}{2}\rceil$ recursively based on the number of vertices.
If $T$ is a star, it is not hard to see that there is  a $\Delta$-completion set of size exactly $\lfloor \frac{n}{2}\rfloor$, see \cref{prop:star}.
Now, assume that $T$ is not a star. Let $v$ be a vertex in $T$ such that all but one neighbor, say $w$, of $v$ are leaves.   
Let $ T' $ be the new tree obtained by contracting $ v$ along with all its neighbors.
%and then contracting the edge incident to this new vertex, i.e.,\ $T'=T \sm (N[v]\sm w)$.
%\todo{$N[v] $ or $N_T(v)$?}
Assume $ T' $ has $ n - n_1 $ vertices, and let $S=N[v]\sm w$. 
The number of vertices of $ T' $ is less than $ |V(T)| $,
so there exists a $\Delta$-completion set $\Delta_{T'}$ of size at most $ \lceil \frac{n - n_1}{2} \rceil $. 
If $n_1\ge 3$, then $S$ is a star with at least 3 vertices and so we apply \Cref{prop:star} and we obtain $\Delta_{S}= \lfloor \frac{n_1}{2} \rfloor$.
Next if $n_1=2$, $S$ is a single edge $uv$, where $u$ is a leaf in $T$.
We now add the edge between the $u$ and $w$ and this case $\Delta_{S}=1$ and we conclude that for every $n_1\ge 1$, we have $\Delta_{S}=\lfloor \frac{n_1}{2} \rfloor$.
The union of $\Delta_{T'}$ and $\Delta_S$ forms a $\Delta$-completion set $\Delta_T$ with size at most $ \lceil \frac{n - n_1}{2} \rceil + \lfloor \frac{n_1}{2} \rfloor $.  
We now show that $ \lceil \frac{n - n_1}{2} \rceil + \lfloor \frac{n_1}{2} \rfloor $ is bounded by $\lceil \frac{n}{2} \rceil$.
\begin{itemize}[leftmargin=*]
\item If $ n_1 $ is an even number, then  
$
    \lfloor \frac{n_1}{2} \rfloor + \lceil \frac{n - n_1}{2} \rceil = \lceil \frac{n}{2} \rceil.$
\item If $ n_1 $ is an odd number:
\begin{itemize}[leftmargin=*]
        \item If $ n $ is even, then  
        $
        \lfloor \frac{n_1}{2} \rfloor + \lceil \frac{n - n_1}{2} \rceil = \lceil \frac{n}{2} \rceil.
        $
        \item If $ n $ is odd, then $ n - n_1 $ is even, so  
        $
        \lceil \frac{n - n_1}{2} \rceil = \frac{n - n_1}{2}.
        $  
        Therefore, we have:  
        $$
        \left\lfloor \frac{n_1}{2} \right\rfloor + \left\lceil \frac{n - n_1}{2} \right\rceil = \left\lfloor \frac{n_1}{2}\right\rfloor + \frac{n - n_1}{2} < \frac{n_1}{2} + \frac{n - n_1}{2} = \left\lceil \frac{n}{2}\right \rceil.
        $$\qedhere
    \end{itemize}
    \end{itemize}
%We note that a longest path can be found in $O(n)$ time for trees by applying a dfs tree.
We observe that the path of order $n$ needs exactly $ \lceil\frac{n}{3}\rceil $ edges for the completion set, as one edge is added for every three consecutive vertices.
\end{proof}
\begin{figure}
    \centering
    \includegraphics[scale=0.7]{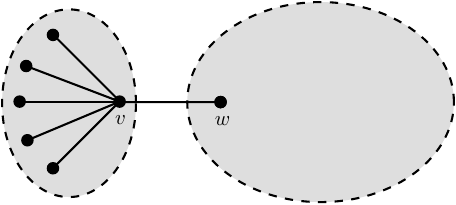}
    \caption{$T'$ is the tree obtained after contracting $N[v]$.}
    \label{fig:cont-t}
\end{figure}

\Cref{thm:tree} allows us to establish an upper bound for an arbitrary connected graph. Specifically, by considering a spanning tree of a graph with order $ n $, we conclude that the upper bound on the size of a $ \Delta $-completion set is at most $ \frac{n}{2} $.  

\begin{corollary}
Let $ n \geq 3 $ be an integer. Then the following holds:
$$\max_{G\text{ connected graph of order }n} \Delta_G \;=\;
   \Bigl\lceil \tfrac{n}{2}\Bigr\rceil . $$
\end{corollary}
Next, we present a linear-time algorithm for finding a minimum $\Delta$-completion set in trees.
We begin with a brief overview of the algorithm.  
We initialize the $\Delta$-completion set $F$ to be empty and set $U$ (the set of uncovered vertices) to include all vertices of the tree.  
We then root the tree arbitrarily in order to define parent–child relationships and compute the depth of each vertex.  
The algorithm proceeds in two main phases: \emph{iterative processing} and \emph{finalization}.  
In the iterative processing phase, we repeatedly select a vertex $v\in U$ of maximum depth and let $u = P(v)$ be its parent.  
If $v$ is the root, we proceed directly to the finalization phase.  
Otherwise, we apply one of the following strategies:  

\begin{itemize}
    \item \textbf{S1 (Sibling pairing):} If $v$ has an uncovered sibling $v'$, 
    add the edge $vv'$ to $F$, and remove $v$, $v'$, and $u$ from $U$.
    
    \item \textbf{S2 (Grandparent connection):} If no uncovered sibling exists, 
    let $g = P(u)$ be the grandparent of $v$.  
    If $g$ exists, add the edge $vg$ to $F$, and remove $v$, $u$, and $g$ from $U$.
    
    \item \textbf{S3 (Root adjustment):} If $u$ is the root (so S2 is not possible), 
    add an edge to cover both $v$ and $u$ by either connecting $v$ to any sibling 
    (even if already covered) or by connecting $u$ to a grandchild $w$.  
    Remove the newly covered vertices from $U$.
\end{itemize}

We now proceed the second step. If the root $r$ remains in $U$ after the main loop, add exactly one edge to cover it (by connecting two of its children or connecting the root to a grandchild).

%\subsubsection{A Linear-Time Algorithm} 

\begin{algorithm}[H]
\caption{Minimum $\Delta$-Completion Set for Trees}
\label{alg:min_delta_tree_dlf_pseudo}
\begin{algorithmic}[1]
\State \textbf{Input:}  Tree $T=(V,E)$ with $|V|\ge 3$ and root $r$
\State \textbf{Output:}  $\Delta$-completion set $F$ of $T$
\State $F \gets \emptyset$; \quad $U \gets V$ \Comment{$U$ is the set of uncovered vertices}
\While{$U \neq \emptyset$}
    \State Select $v\in U$ of maximum depth
    \If{$v \neq r$} 
        \State $u \gets \textit{parent}[v]$ and $C_u \gets \textit{children}[u] \cap U$
        \While{$|C_u| \ge 2$}
            \State Remove any two vertices $v_1,v_2$ from $C_u$ and  add edge $v_1v_2$ to $F$
            \State $U \gets U \sm \{v_1,v_2,u\}$
        \EndWhile
        \If{$C_u = \{v\}$}
            \If{$u \neq r$} \Comment{Grandparent connection}
                \State $g \gets \textit{parent}[u]$ and add edge $vg$ to $F$
                \State $U \gets U \sm \{v,u,g\}$
            \Else \Comment{$u$ is the root}
                \State Choose any child $x$ of $u$ or $v$ and add edge $ux$ or $vx$ to $F$
                \State $U \gets U \sm \{u,v,x\}$
            \EndIf
        \EndIf
    \Else \Comment{$v$ is the root}
        \If{$v$ has at least two children $c_1,c_2$}
            \State Add edge $c_1c_2$ to $F$
            \State $U \gets U \sm \{v,c_1,c_2\}$
        \Else \Comment{$v$ has exactly one child $c_1$}
            \State Let $w$ be a grandchild of $v$ and add edge $vw$ to $F$
            \State $U \gets U \sm \{v,c_1,w\}$
        \EndIf
    \EndIf
\EndWhile
\end{algorithmic}
\end{algorithm}

\begin{observation}\label{obs:cudisjoint}
The sets $C_u$ are disjoint. 
\end{observation}

\begin{theorem}\label{thm:alg-trees}
Let $T$ be a tree with  at least three vertices. \Cref{alg:min_delta_tree_dlf_pseudo} computes a minimum $\Delta$-completion set for $T$ in linear time.
\end{theorem}

\begin{proof}
Let $F_{\mathrm{ALG}}$ be the $\Delta$-completion set produced by the algorithm, 
and let $F_{\mathrm{OPT}}$ be a minimum $\Delta$-completion set. 
By \cref{cor:lower_bound_for_subset}, we may assume there exists an optimal solution 
(which we also denote by $F_{\mathrm{OPT}}$) such that every edge in 
$F_{\mathrm{OPT}}$ connects vertices at distance~2 in $T$.
We now use a charging argument to show that 
\[
|F_{\mathrm{ALG}}| \;\leq\; |F_{\mathrm{OPT}}|.
\]
Specifically, we will charge the cost of the edges added by $F_{\mathrm{ALG}}$ 
to the edges in $F_{\mathrm{OPT}}$.
The algorithm identifies $C_u$, the set of children of $u$ that are currently marked as uncovered.
We charge the vertex $u$, denoted $\mathrm{ch}(u)$, 
the number of the edges added to cover $C_u$. 
We note that this charge arises from \emph{Sibling Pairing} and from handling the remaining uncovered child.
More precisely we obtain, 
\[
\mathrm{ch}(u) \;=\; \Big\lceil \tfrac{|C_u|}{2} \Big\rceil.
\]
We also define the \emph{finalization charge}, $\mathrm{ch}(\mathrm{Final})$, 
which equals $1$ if the root $r$ remains uncovered after the traversal and $0$ otherwise.
The total size of the algorithm's solution is
\[
|F_{\mathrm{ALG}}| \;=\; \sum_{u \in V} \mathrm{ch}(u) \;+\; \mathrm{ch}(\mathrm{Final}).
\]
Let $F_{\mathrm{OPT}}$ be a minimum $\Delta$-completion set. 
By \Cref{cor:lower_bound_for_subset}, we may assume every edge in $F_{\mathrm{OPT}}$ 
connects two vertices at distance~2 in $T$.
For each vertex $u$, define $E_{\mathrm{OPT}}(C_u)$ as the minimum subset 
of edges from $F_{\mathrm{OPT}}$ required to cover all vertices in $C_u$.  
Since $C_u$ consists of siblings in the tree, it forms an independent set.  
An edge at distance~2 can cover at most two vertices in $C_u$, 
so we obtain the lower bound
\[
|E_{\mathrm{OPT}}(C_u)| \;\geq\; \Big\lceil \tfrac{|C_u|}{2} \Big\rceil \;=\; \mathrm{ch}(u).
\]

\begin{claim}
For any two distinct vertices $u_1$ and $u_2$, the sets 
$E_{\mathrm{OPT}}(C_{u_1})$ and $E_{\mathrm{OPT}}(C_{u_2})$ are disjoint.
\end{claim}

\begin{proof}
Suppose, for contradiction, that there exists an edge 
$e \in F_{\mathrm{OPT}}$ with $e \in E_{\mathrm{OPT}}(C_{u_1}) \cap E_{\mathrm{OPT}}(C_{u_2})$. 
Then $e$ covers some $v_1 \in C_{u_1}$ and some $v_2 \in C_{u_2}$. 
Since $C_{u_1}$ and $C_{u_2}$ are disjoint see \Cref{obs:cudisjoint}, we have $v_1 \neq v_2$.  
Thus $e=v_1v_2$, and by \Cref{cor:lower_bound_for_subset}, $\mathsf{dist}_T(v_1,v_2)=2$.

Consider the relative positions of $u_1 = P(v_1)$ and $u_2 = P(v_2)$:
\begin{itemize}
    \item \textbf{Case 1:} $u_1$ and $u_2$ are siblings.  
    Then the path is $v_1 - u_1 - P(u_1) - u_2 - v_2$, 
    so $\mathsf{dist}_T(v_1,v_2)=4$.
    \item \textbf{Case 2:} One is an ancestor of the other.  
    Assume $u_2$ is an ancestor of $u_1$.  
    If $u_2 = P(u_1)$, then the path is $v_1 - u_1 - u_2 - v_2$, 
    giving $\mathsf{dist}_T(v_1,v_2)=3$.  
    If $u_2$ is a higher ancestor, then $\mathsf{dist}_T(v_1,v_2) > 3$.
    \item \textbf{Case 3:} $u_1$ and $u_2$ are otherwise unrelated.
    Then $\mathsf{dist}_T(v_1,v_2) > 4$.
\end{itemize}
In all cases, $\mathsf{dist}_T(v_1,v_2) > 2$, contradicting \Cref{cor:lower_bound_for_subset}.  
Therefore, $E_{\mathrm{OPT}}(C_{u_1}) \cap E_{\mathrm{OPT}}(C_{u_2}) = \emptyset$.
\end{proof}
If $\mathrm{ch}(\mathrm{Final})=1$, then the root $r$ was never covered during the traversal. This implies:
\begin{enumerate}
    \item When processing $r$, we had $|C_r|=0$.
    \item When processing any child $u$ of $r$, S2 (Grandparent Connection to $r$) was not executed.
\end{enumerate}
Thus $F_{\mathrm{OPT}}$ must also cover $r$ using some edge $e_r \in F_{\mathrm{OPT}}$.  
We claim that $e_r \notin E_{\mathrm{OPT}}(C_u)$ for any $u$.  
Indeed, if $e_r \in E_{\mathrm{OPT}}(C_u)$, then $e_r$ would cover both $r$ 
and some $v \in C_u$, implying that $r$ acted as the parent or grandparent. But in such a case, the algorithm would also have covered $r$, contradicting the assumption $\mathrm{ch}(\mathrm{Final})=1$.

Therefore, $e_r$ has not yet been charged, and we may charge 
$\mathrm{ch}(\mathrm{Final})$ to $e_r$.  
The disjointness of charges is preserved.
Since the sets $E_{\mathrm{OPT}}(C_u)$ are disjoint and 
the finalization cost is also accounted for, we obtain
\[
|F_{\mathrm{OPT}}| \;\geq\; \sum_{u \in V} \mathrm{ch}(u) \;+\; \mathrm{ch}(\mathrm{Final})
\;=\; |F_{\mathrm{ALG}}|.
\]
As $F_{\mathrm{OPT}}$ is optimal, this shows $|F_{\mathrm{ALG}}| = |F_{\mathrm{OPT}}|$, 
and therefore the algorithm is optimal.
\end{proof}

\subsection{Chordal Graphs}
In this section, we study the problem for connected chordal graphs of order at least $3$. 
Recall that a graph is  \defin{chordal} if every induced cycle in the graph has exactly three vertices. 
Before proceeding with the proofs, we introduce some necessary definitions and notation.

\begin{definition}
Let $G$ be a chordal graph. 
\begin{itemize}
\item Define $\mathcal{T}(G)$
%\todo{do we use this notation for tree decomposition somewhere else?} 
as the set of all maximal subgraphs of $G$ induced by its bridges. 
\item Each element of $\mathcal{T}(G)$ contains two types of vertices: those that belong to a clique of order at least $3$ and those that are cut vertices. 
We refer to the former as \defin{clique vertices}.  
Let $T \in \mathcal{T}(G)$ and let $u \in V(T)$.  
We say that $v$ is a \defin{pendant clique vertex} (or for simplicity pcv) if $v$ belongs to a clique and has degree exactly one in $T$.  
For each $T \in \mathcal{T}(G)$, let $T'$ be the forest obtained by deleting all pendant clique vertices from $T$ and set $\mathcal T'(G)=\{T'\mid T\in \mathcal T(G)\}$. 
%\todo{Anil: An illustration will be nice.}
An edge  $uv$  is called a \defin{clique edge} if both  $u$  and  $v$  are clique vertices.
\item Let $F$ be a $\Delta$-completion set of $G$. Then an edge $uv\in F$ is called \defin{crossing}, if there are two distinct trees $T_i$ and $T_j$ of $\mathcal T(G)$ such that $u\in V(T_i)$ and $v\in V(T_j)$.
We denote the set of all crossing edges between $T_i$ and $T_j$ by $E_{ij}$.
\item If the size of  $T $ is at most two, we adopt the convention that $ \Delta_T = 1$.
\end{itemize} 
\end{definition}
\begin{figure}
    \centering
    \includegraphics[scale=0.7]{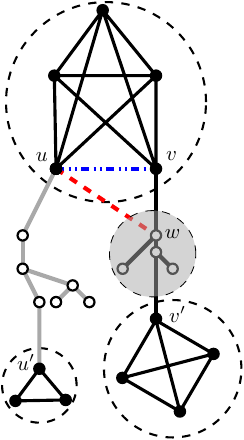}
    \caption{The vertex $u$ is a clique vertex, and $v$ is a pendant clique vertex. The subtree indicated in the gray area is obtained by deleting the pendant clique vertices $v$ and $v'$. The edge 
$uw$ is a crossing edge.}
    \label{fig:placeholder}
\end{figure}

In order to study $\Delta$-completions of chordal graphs, we first show that there exists a $\Delta$-completion set of $G$ that contains no crossing edges between distinct trees of $\mathcal{T}(G)$.  
More precisely, 
%we provide a procedure that takes an arbitrary $\Delta$-completion set $F$ of size $\Delta_G$ as input and produces a $\Delta$-completion set $F'$ of the same size that does not contain any crossing edges.
let $F$ be a $\Delta$-completion of $G$. Let $T_i'$ and $T_j'$ be two distinct trees of $\mathcal T'(G)$ of order at least $3$.
In the next lemma, we show that there is a way to distribute the edges $E_{ij}$(the crossing edges between $T_i$ and $T_j$) within $T_i'$ and $T_j'$ so that we get another $\Delta$-completion $F'$ of $G$ such that $|F'|\leq |F|$ and repeat this procedure between any two distinct trees in $\mathcal T(G)$.
We note that if $ uv $ is a crossing edge between $ T_i $ and $ T_j $, then at least one of $ u $ or $ v $ must be a clique vertex.  
To see this, consider coloring all edges of $ G $ blue and all edges of $ F $ red. 
By \cref{lem:lower_bound}, we know that there exists a vertex $z$ in $N_G(u)\cap N_G(v)$.
This implies that either $u$ or $v$ should be a clique vertex.
To summarize, one of the following cases only happens, as a crossing edge between $T_i$ and $T_j$.

\begin{itemize}
    \item $v$ is not a pcv and $u$ is not a pvc, see \Cref{fig:tri_1}(a).
    \item $v$ is a pcv and $u$ is not a pvc, and the crossing edge is incident to $v$, see \Cref{fig:tri_1}(b).
    \item $v$ is a pcv and $u$ is not a pvc, and the crossing edge is incident to $u$, see \Cref{fig:tri_1}(c).   
    \item $v$ is  a pcv and $u$ is a pvc, see \Cref{fig:tri_1}(d).
\end{itemize}
\begin{figure}[H]
    \centering
    \subfloat[\centering ]{{\includegraphics[scale=0.65]{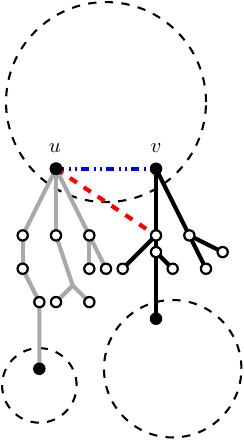}}}
\qquad
\subfloat[\centering .]{{\includegraphics[scale=0.65]{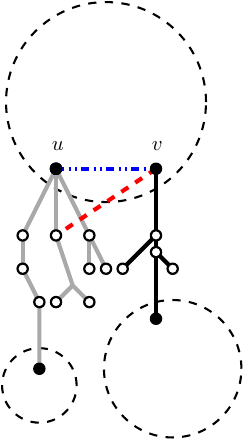} }}%
\qquad
\subfloat[\centering  ]{{\includegraphics[scale=0.65]{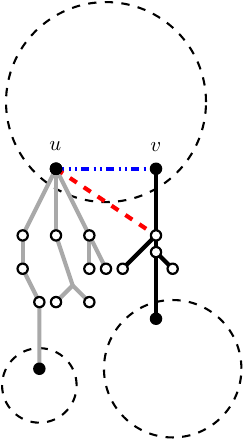} }}%
\qquad
\subfloat[\centering ]{{\includegraphics[scale=0.65]{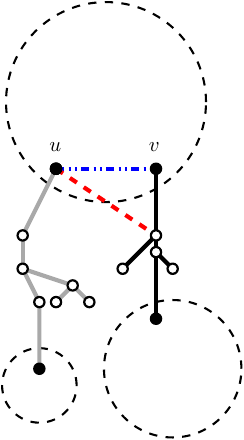} }}%
\label{fig:tri_1}
\caption{Each dashed blob represents a clique of order at least three, and the blue( dashed-dotted) edge $ uv $ is an edge within one of these cliques and the red( dashed) edges are from $\Delta$-completion.
In this example, the set $ \mathcal{T}(G) $ consists of two trees: one shown in gray and the other in black. 
The black vertices represent the clique vertices.}
\end{figure}

\begin{lemma}\label{lem:transfering}
Let $ G $ be a chordal graph, and let $ F $ be a $ \Delta $-completion set of size $ \Delta_G $.  
Then there exists another $ \Delta $-completion set $ F' $ of size $\Delta_G$, obtained by replacing each crossing edge with an edge within some tree of $\mathcal{T'}(G)$.
\end{lemma}

\begin{proof}
Let $F$ be a minimum $\Delta$-completion set obtained by \cref{lem:lower_bound}.
We note that if $ xy\in F $ is a crossing edge between $ T_i $ and $ T_j $, then at least one of $ x $ or $ y $ must be a clique vertex.  
By \cref{lem:lower_bound}, there exists a vertex $ z \in N_G(x) \cap N_G(y) $.  
This implies that either $ x $ or $ y $ must be clique vertices.  
Without loss of generality, assume that $ y $ is a clique vertex.
We note that $x$ is not a clique vertex; otherwise, $xy$ is not crossing.
Assume that $x$ lies in the tree $ T_i $.  
Let $ C_i $ be the component of $ T_i' $ that contains $ x $.  
We split the proof into two cases:
\begin{itemize}
\item If $ |V(C_i)| \geq 3 $, then since $ C_i $ has at least three vertices, we can find a path $ p $ of three vertices that includes $ x $. Now, we replace the edge $ xy $ in $ F $ with a non-edge $ f' $ such that adding $ f' $ to $ p $ forms a new triangle.  
%\todo{Anil: An illustration, please.}
\item Let $ |V(C_i)| < 3 $, we note that in this case, $ \Delta_{C_i} = 1 $. Thus, we use the clique vertices to replace $ f $ with an edge that creates a triangle containing $x$, see \cref{fig:small_trees}.
\end{itemize}
\begin{figure}[H]
\centering
\subfloat[\centering ]{{\includegraphics[scale=0.60]{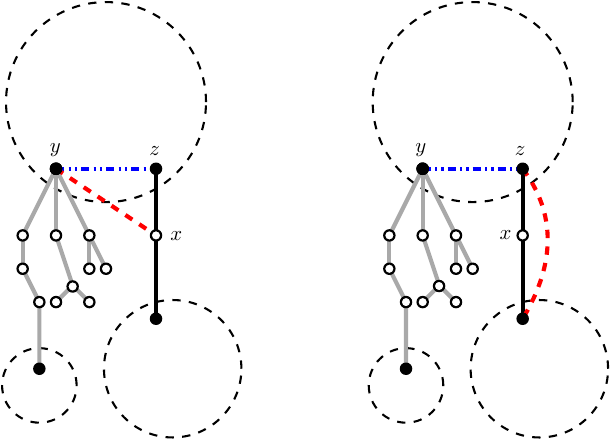}}}
\qquad
\subfloat[\centering .]{{\includegraphics[scale=0.60]{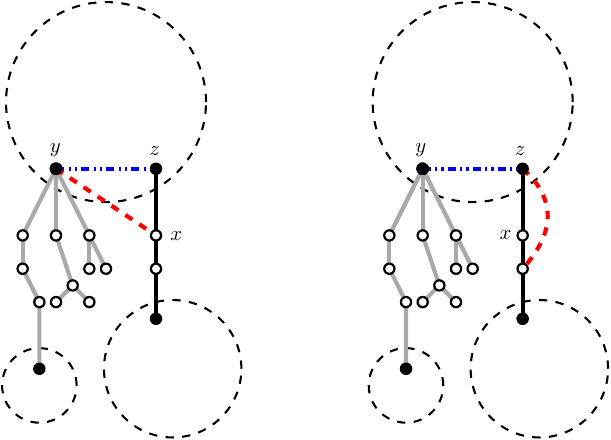} }}

\caption{The tree with black edges is $ T_i $.  (a) illustrates the case where $ T_i' $ consists of a single vertex, while  (b) represents the case where $ T_i' $ consists of two vertices. In each figure, we use a clique vertex to cover the vertex $ x $.}\label{fig:small_trees}
\end{figure}
%\todo{within the caption, we can just say (a) represents this and (b) represents that instead of Figure (a)}
We repeat the procedure for all crossing edges.
\end{proof}

%\todo{some problems with this proof. What if we don't have any pendant clique vertices? then $T'_i=T_i$. Also, $T'_i$ could have multiple components. what component are we working with? }
%The proof of the following theorem is deferred to the appendix. 
\begin{theorem}\label{thm:chordal}
Let $G$ be a chordal graph and let $\mathcal T'(G)=\{T_1',\ldots,T_m'\}$. Then $\Delta_G= \sum_{i=1}^m \Delta_{T_i'} $ 
\end{theorem}

\begin{proof}
It is clear that by adding a minimum $\Delta$-completion set for each tree $T_i'$, we obtain a $\Delta$-completion set of $G$, which implies that  $\Delta_G \leq \sum_{i=1}^{m} \Delta_{T_i'}$. 
We note that when the order of $T'$ is at most 2, we have $\Delta_{T'} = 1$. In this case, we use the clique vertices of $T$ to add this additional edge.
By applying \cref{lem:transfering}, we can replace each crossing edge between trees in $\mathcal{T}(G)$  with an edge within a tree in $\mathcal{T}'(G)$.
This proves that $\Delta_G\geq \sum_{i=1}^{m} \Delta_{T_i'}$.
\end{proof}
\subsubsection{Algorithm}
Combining~\Cref{thm:chordal} with~\Cref{thm:alg-trees}, taking into account that a chordal graph can be decomposed into maximal bridge subgraphs in linear time (with respect to its order and size).
More precisely,
By Tarjan’s linear-time bridge algorithm, an undirected (in particular, chordal) graph can be decomposed into its maximal bridgeless (2-edge-connected) subgraphs in overall $\mathcal O(n+m)$ time, \cite{Tarjan1974Bridges}.
Thus we can immediately conclude with the following.

\begin{corollary}
Let $G$ be a chordal graph of order $n$ at least $3$, with $m$ edges. 
Then, the minimum $\Delta$-completion set of $G$ can be found in $\mathcal{O}(n+m)$ time.
In addition, the triangle-covered problem in chordal graphs can be solved in $\mathcal O(n+m)$ time.
\end{corollary}

%\todo{anil: Shall we drop cacti from the conference version?}
\subsection{Cacti}
In this section, we study the minimum $ \Delta $-completion sets of cactus graphs.
Recall that a \defin{cactus graph} is  a connected graph whose every block\footnote{A block is a maximal connected subgraph of a given graph $G$ that has no cut-vertex.} is either a single edge or a cycle. 
In terms of tree decompositions, every cactus graph admits a tree decomposition in which each bag induces either a cycle or a single edge.
We note that if $G$ is $2$-edge-connected, then every bag should be a cycle.

We note that for every $n$, there exists a $2$-edge-connected cactus graph of order $2n+1$ that is triangle-covered. This can be constructed by taking $n$ triangles and identifying all of them at a single common vertex.
On the other hand, for every $n$, there exists a $2$-edge-connected cactus graph $G$ of order $2n$ such that $\Delta_G = 1$. To see this, recall that from the previous construction, there exists a $2$-edge-connected cactus graph of order $2n - 3$ that is triangle-covered. We then take a cycle of order $4$ and identify one of its vertices with the vertex of maximum degree in $G$, yielding a new graph of order $2n$ with $\Delta_G = 1$.
In the next theorem, we show that the upper bound for the minimum $\Delta$-completion set of a 2-edge-connected cactus graph of order $n$ is $\frac{2n}{5}$.

\begin{theorem}\label{thm:2_conn_cactus}
Let $ n \geq 4 $ be an integer. Then, the following holds:
$$\mm\left(\left\{ \Delta_G \mid G \text{ is a $2$-edge-connected cactus graph of order } n \right\}\right)= \left\lceil \frac{2n}{5} \right\rceil $$
\end{theorem}

\begin{proof}
We proceed by induction on the number of vertices.
If $n=4$, then $G$ is a cycle with four vertices, and we are done by adding one single edge.
So we assume that $n\geq 4$.
Since $ G $ is $ 2 $-edge-connected, there exists a tree decomposition $ (T, \mathcal{V}) $ of $ G $ such that each bag corresponds to a cycle.  
Consider a leaf of $ T $, and let $ C $ be the cycle associated with this leaf, containing $ n_1 $ vertices. Let $ v \in V(C) $ belong to at least two cycles in $G$.
%\todo{I suggest we remove also and write it as let $v\in V(C)$ be the unique vertex that is in two cycles in $G$. } Without loss of generality, we assume that $ C $ has at least four vertices.  
Define $ G' $ as the graph obtained by removing all vertices in $ V(C) \sm \{ v \} $ i.e., $G'= G\sm (V(C)\sm\{v\})$.  
Let $ P $ be the path obtained from $ C $ by deleting the vertex $ v $, and
let $ E $ be a $\Delta$-completion set of $ P $ of size $ \lceil \frac{n_1 - 1}{3} \rceil $.  
By the induction hypothesis, $ G' $ has a $\Delta$-completion set $ E' $ of size at most $ \frac{2(n - n_1 + 1)}{5} $.  
It follows that $ E \cup E' $ forms a $\Delta$-completion set of $ G $ of size at most  
$$
\frac{2(n - n_1 + 1)}{5} + \frac{n_1 - 1}{3} = \frac{6n - n_1 + 1}{15} \leq \frac{2n}{5}.
$$
Next, we construct an infinite family of graphs for which $ \Delta_G = \frac{2n}{5} $ holds.  
Consider a cycle $ C $ of length $ n $, and attach a cycle of length $5$ to each of its vertices.  
We denote the resulting graph by $ G_n $.  
Observe that $ G_n $ has $ 5n $ vertices.  
We show that $\Delta_{G_n} = 2n$.

\begin{figure}[H]
    \centering
    \subfloat[\centering An illustration for $G_5$.]{{\includegraphics[scale=0.54]{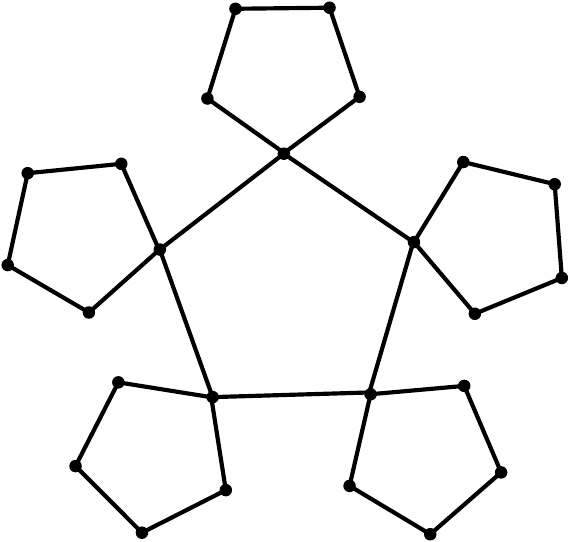}}}
\qquad
\subfloat[\centering The paths $ \pi_i $ are highlighted in gray.]{{\includegraphics[scale=0.54]{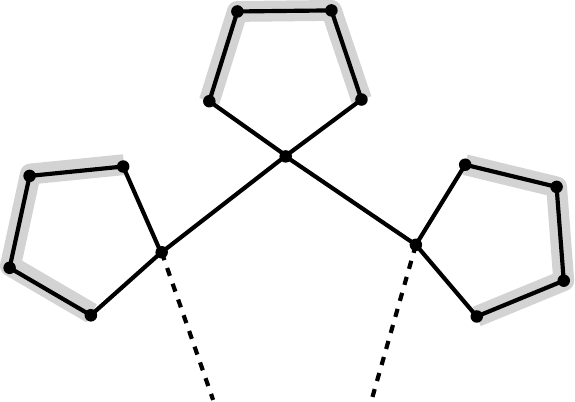} }}%
    \caption{}
    \label{fig:triangle_1}
\end{figure}
Let $ C_1, \ldots, C_n $ be the leaf cycles of $ G_n $, and let $ C $ be the central cycle.  
For each $ i \in [n] $, let $ \pi_i $ be the path of length $ n-1 $ that contains all vertices of $ C_i $ exceptthe shared vertex $ C $.  
We consider each  path $ \pi_i $ as the subgraph $ H_i $. 
We note that each $ \pi_i $ requires at least 2 additional edges.  
Taking into account \cref{dist=3}, this results in a total of at least $ 2n $ edges, as required.
\end{proof}

%%%%%%%%%%%%%%%%%%%%%%%%%%%%%%%%%%%%%%%%%%%%%%%%%%%%%%%%%%%%%%%%%%%%%%%%%%%%%%%%%%%%%%%%%%%%%%%%%%%%%%%%%%%%%%%%%%%%%%%%%%%%%%%%%%%%%%%%%%%%%%%%%%%%%%%%%%%%%%%%%%%%%%%
\section{Random graphs}
A random graph is obtained by starting with a set of $n$ isolated vertices and adding successive edges between them at random. This model is called the \textit{Erdős–Rényi} model, denoted by $\mathbb G(n,p)$. 
More precisely, in $\mathbb G(n,p)$, every possible edge occurs independently with probability $0 < p < 1$. 
Recall that a vertex is \defin{unsaturated} if it does not belong to any triangle.
We aim to find the expectation of the number of unsaturated vertices, i.e 
$$\mathbb E(\# \textit{ of unsaturated vertices})=\sum_{v} \mathbb P(v \textit{ is unsaturated})=n\mathbb P(v \textit{ is unsaturated})$$
%Let $X_{uvw}$ denote the indicator function for the triangle $uvw$, and define 
%$X=\sum_{v,w} X_{uvw}$.
%In addition, we set 
%\begin{equation}\label{mu_delta}
%\mu\coloneqq\sum X_{uvw}=\binom{n-1}{2}p^3, \Delta \coloneqq \sum \mathbb E(X_{uvw}=1=X_{uvw'})=\binom{n-2}{2}(n-3)p^5
%\end{equation}
We also need Janson's inequality:
\begin{lemma}{\rm \cite[Theorem 3.4]{MR2547432}}(Janson's inequality:) Let $X=\sum_A X_A$ and let $\mu=E(X)=\sum_A \mathbb P[X_A=1]$.
Let $$\Delta=\sum_{A\sim B}E(X_AX_B)=\sum_{A\sim B}\mathbb P[X_A=1=X_B]$$
Then we have 
$$\mathbb P[X=0]\leq \exp({-\frac{\mu^2}{\mu+\Delta}})$$
\end{lemma}

\begin{lemma}
Let $ G = \mathbb G(n, p) $ be a random graph.
If $ p \gg n^{-2/3}$, then $$\lim_{n\to \infty}\mathbb P(G \textit{ has an unsaturated vertex})=0$$
\end{lemma}
%\todo{Anil: unsaturated instead of saturated!}

\begin{proof}  
Fix a vertex $ u $. We compute the probability that $ u $ is unsaturated, i.e., $ u $ does not participate in any triangle. Let $ X_{uvw} $ be the indicator random variable for the triangle $ uvw $, and define:  
$X = \sum_{\substack{v, w \in V(G) \\ v \neq w \neq u}} X_{uvw}$.
Then, the probability that $ u $ is unsaturated is given by  
$\mathbb{P}(u \text{ is unsaturated}) = \mathbb{P}(X = 0)$.
The expected number of triangles containing $ u $ is:  
\[
\mu = \mathbb{E}[X] = \binom{n-1}{2} p^3 \approx \frac{n^2}{2} p^3.
\]  
The term $ \Delta $ accounts for pairs of triangles sharing an edge incident to $ u $. For each edge $ uv $, there are $ \binom{n-2}{2} $ pairs of triangles sharing $ uv $. Summing over all $ n-1 $ edges incident to $ u $, we obtain:  
\[
\Delta = \sum_{v, w, w'} \mathbb{E}[X_{uvw} X_{uvw'}] = (n-1) \binom{n-2}{2} p^5 \approx \frac{n^3}{2} p^5.
\]  
Substituting $ \mu $ and $ \Delta $:  
\[
\frac{\mu^2}{\mu + \Delta} \approx \frac{\left(\frac{n^2}{2} p^3\right)^2}{\frac{n^2}{2} p^3 + \frac{n^3}{2} p^5} = \frac{\frac{n^4}{4} p^6}{\frac{n^2}{2} p^3 (1 + np^2)} = \frac{n^2 p^3}{2(1 + np^2)}.
\]  
Thus,  
\[
\mathbb{P}(X = 0) \leq \exp\left(-\frac{n^2 p^3}{2(1 + np^2)}\right).
\]  
We analyze two cases based on the behavior of $ np^2 $:
\begin{itemize}
    \item \textbf{Case 1:} If $ np^2 \to 0 $, then $ 1 + np^2 \approx 1 $, so  
    \[
    \frac{n^2 p^3}{2(1 + np^2)} \approx \frac{n^2 p^3}{2}.
    \]  
    For $ p \gg n^{-2/3} $, we have $ n^2 p^3 \to \infty $, ensuring that $ \mathbb{P}(X = 0) \to 0 $.  

    \item \textbf{Case 2:} If $ np^2 \to \infty $, then $ 1 + np^2 \approx np^2 $, so  
    \[
    \frac{n^2 p^3}{2(1 + np^2)} \approx \frac{n^2 p^3}{2 np^2} = \frac{n p}{2}.
    \]  
    If $ p \gg n^{-1} $, then $ np \to \infty $, leading again to $ \mathbb{P}(X = 0) \to 0 $.  
\end{itemize}
Since $ p \gg n^{-2/3} $, we conclude that  $\mathbb{P}(X = 0) = o\left(\frac{1}{n}\right)$.
Thus, the expected number of unsaturated vertices is  
\[
\mathbb{E}[\# \text{ unsaturated vertices}] = n \cdot \mathbb{P}(X = 0) \leq n \cdot \exp\left(-\frac{n^2 p^3}{2(1 + np^2)}\right) \to 0.
\]  
This completes the proof.  
\end{proof}  

An event $ E $ is \defin{decreasing} if, whenever $ E $ occurs for a graph $ G $, it also occurs for any subgraph $ G' \subseteq G $ (i.e., removing edges from $ G $ cannot destroy $ E $).

\begin{lemma}[FKG inequality]\label{lem:fkg}
If $A$  and  $B$ are two decreasing events, then
$$\mathbb{P}(A \cap B) \geq \mathbb{P}(A)\mathbb{P}(B).$$
\end{lemma}
\begin{lemma}
If the expected number of unsaturated vertices in $\mathbb G(n, p)$ satisfies  
\[
\lim_{n\to\infty}\mathbb{E}[\text{\# unsaturated vertices}] = 0,
\]
then it must hold that $ p \gg n^{-2/3} $.
\end{lemma}
\begin{proof}
Suppose for contradiction that $ p \not\gg n^{-2/3} $, i.e., either  $p = o(n^{-2/3})$ or  $p = \Theta(n^{-2/3})$.
We will show that in both cases,  $\mathbb{E}[\text{\# unsaturated vertices}] \not\to 0$.\\ 
\textbf{Case 1:} let $ p = o(n^{-2/3}) $.  
The expected number of triangles containing a fixed vertex $ u $ is  
$\mu = \binom{n-1}{2} p^3 \approx \frac{n^2}{2} p^3 = o(n^2 \cdot n^{-2}) = o(1)$.  
We note that the probability of $u$ lies in at least one triangle is $\mathbb P(\cup X_{uvw})$ which is less than $\sum \mathbb P(X_{uvw})=\mu$.
Thus $\mathbb P(u \text{ is unsaturated}) \ge 1 - \mu$ and so the expected number of unsaturated vertices is  
\[
\mathbb{E}[\text{\# unsaturated vertices}] \ge n(1 - \mu) \approx n \to \infty.
\]  
\textbf{Case 2:} 
Suppose \( p = \Theta(n^{-2/3}) \).
Let \( X \) be the number of triangles containing \( u \).  
The probability that \( u \) is \emph{unsaturated} is \( \mathbb{P}(X = 0) \).  
We note that the events  $A_{vw} = \{\text{triangle } uvw \text{ is absent} \}$ are \emph{decreasing events} for all pairs \( v, w \).
Since the absence of a triangle corresponds to the absence of an edge and the edge indicators are independent, these events are decreasing.
By the FKG inequality, see \cref{lem:fkg} for decreasing events , we have:
\[
\mathbb{P}(X=0)=\mathbb{P} \left( \bigcap_{v, w} A_{vw} \right) \geq \prod_{v, w\in V\sm \{u\}} \mathbb{P}(A_{vw}) = \prod_{v, w} (1 - p^3),
\]
where the product is over all unordered pairs \( \{v, w\} \subseteq V \sm \{u\} \).
Since \( p^3 = \Theta(n^{-2}) \), we use the inequality
$\log(1 - x) \geq -x - x^2 \quad \text{for } x < 0.5$. Then we obtain the following:
\[
\log \left( \prod_{\{v, w\}} (1 - p^3) \right) 
\geq -\sum_{\{v, w\}} (p^3 + p^6) = -\binom{n - 1}{2} (p^3 + p^6).
\]

Next by exponentiating both sides, we obtain:
\[
\prod_{\{v, w\}} (1 - p^3) 
\geq \exp \left( -\binom{n - 1}{2} p^3 - \binom{n - 1}{2} p^6 \right).
\]

Now observe that:
\[
\binom{n - 1}{2} p^6 = \Theta(n^2) \cdot \Theta(n^{-4}) = \Theta(n^{-2}) \textit{ and } \binom{n - 1}{2} p^3 = \Theta(n^2) \cdot \Theta(n^{-2}) = \Theta(1).
\]

So we conclude that $\mathbb{P}(X = 0) \geq e^{-\Theta(1)-\Theta(n^{-2})}$ and there are constants $c_1,c_2$ such that $\mathbb{P}(X = 0) \geq e^{-c_1-c_2n^{-2}}$
Thus, the expected number of unsaturated vertices satisfies:
\[
\mathbb{E}[\# \text{ unsaturated vertices}] \geq n \cdot e^{-c_1-c_2n^{-2}}  \to \infty.
\]
This contradicts the assumption that
$\lim_{n \to \infty} \mathbb{E}[\# \text{ unsaturated vertices}] = 0$.
\end{proof}

\section{Conclusion and Problems}

In this paper, we introduced a new edge modification problem, the $\Delta$-completion problem, which involves transforming a graph into a triangle-covered graph. We showed that the problem is $\mathbb{NP}$-complete and does not admit a constant-factor approximation algorithm; it remains $\np$-complete even when restricted to bipartite connected graphs. 
Next, we established tight bounds on the minimum $\Delta$-completion size and developed algorithms for trees and chordal graphs.
We close the paper with the following open questions.

\medskip
Let $ G = (A, B) $ be a complete bipartite graph. By adding a single edge within each part $ A $ and $ B $, one can obtain a triangle-covered graph. This observation naturally leads to the following question:

\medskip
\noindent{\bf Problem 1.} Let $ G $ be a bipartite graph with bounded degree. Is there a polynomial-time algorithm to compute a minimum $ \Delta $-completion set of $ G $?
\medskip

We have shown that the upper bound for the $\Delta$-completion number of $2$-edge-connected cactus graphs of order~$n$ is $2n/5$, and that this bound is tight. We believe that using dynamic programming, as seen for example in~\cite {DasK12}, one can also design an optimal  algorithm for finding minimum $\Delta$-completion sets in $2$-edge-connected cacti.

\medskip
\noindent\textbf{Problem 2.} Design an optimal algorithm to compute a minimum $\Delta$-completion set of $2$-edge-connected cactus graphs.

\section*{Acknowledgments}
This research was initiated during the Banff workshop ‘Movement and Symmetry in Graphs (24w5298)’.

\bibliographystyle{plainurlnat}
\bibliography{ref.bib}

\end{document}